\documentclass[11pt]{article}
\usepackage{amssymb}
\usepackage{amsfonts}
\usepackage{amsmath}
\usepackage{amsmath}
\usepackage{amsfonts}
\usepackage{amssymb}
\usepackage{setspace}
\usepackage{enumitem}
\usepackage{graphicx}
\usepackage{longtable}
\usepackage{float}

\usepackage[strict]{changepage}
\usepackage{afterpage}

\usepackage{chngcntr}
\usepackage{apptools}
\usepackage[bottom]{footmisc}
\usepackage[round]{natbib}   % omit 'round' option if you prefer square brackets
\usepackage{color}
\usepackage{algorithm}
\usepackage{algorithmic}
\usepackage{mathtools}

\DeclarePairedDelimiter\floor{\lfloor}{\rfloor}

\newtheorem{theorem}{Theorem}

\newenvironment{proof}[1][Proof]{\noindent \textbf{#1.} }{\  \rule{0.5em}{0.5em}}

\title{\vspace{-4cm}Continuous-Time Mean--Variance Portfolio Selection: A  Reinforcement Learning Framework\thanks{We are grateful for comments from the seminar participants at the Fields Institute.
Wang gratefully acknowledges financial supports through the FDT Center for Intelligent Asset Management at Columbia. Zhou gratefully acknowledges financial supports through a start-up grant at Columbia University and through the FDT Center for Intelligent Asset Management.}}

\author{Haoran Wang\footnote{Department of Industrial Engineering and Operations Research, Columbia University, New York, NY 10027, USA. Email: hw2718@columbia.edu.}\and Xun Yu Zhou\footnote{Department of Industrial Engineering and Operations Research, and The Data Science Institute, Columbia University, New York, NY 10027, USA. Email: xz2574@columbia.edu.}}
\date{\vspace{-5ex}}

\begin{document}
\maketitle
\begin{center}

First draft: February 2019\\
This version: May 2019
\end{center}

\bigskip
\bigskip

\begin{abstract}

We approach the continuous-time mean--variance (MV) portfolio selection with reinforcement learning (RL). The problem is to achieve the best tradeoff between exploration and exploitation, and is formulated as an entropy-regularized, relaxed stochastic control problem. We prove that the optimal feedback policy for this problem must be Gaussian, with time-decaying variance. We then establish connections between the entropy-regularized MV and the classical MV, including the solvability equivalence and the convergence as exploration weighting parameter decays to zero. Finally, we prove a policy improvement theorem, based on which we devise an implementable RL algorithm. We find that our algorithm  outperforms both an adaptive control based method and a deep neural networks based algorithm by a large margin in our simulations.

%We consider continuous-time Mean--variance (MV) portfolio optimization problem in the Reinforcement Learning (RL) setting. The problem falls into the entropy-regularized relaxed stochastic control framework recently introduced in Wang et al. (2019). We derive the feedback exploration policy as the Gaussian distribution, with time-decaying variance. Close connections between the entropy-regularized MV and the classical MV are also discussed, including the solvability equivalence and the convergence as exploration decays. Finally, we prove a policy improvement theorem (PIT) for the continuous-time MV problem under both entropy regularization and control relaxation. The PIT leads to an implementable RL algorithm for the continuous-time MV problem. Our algorithm outperforms an adaptive control based method that estimates the underlying parameters in real-time and a state of art RL method that uses deep neural networks for continuous control problems by a large margin in nearly all simulations.

 \bigskip
 {\bf Key words.} Reinforcement learning, mean-variance portfolio selection,
 entropy regularization,  stochastic control, value function, Gaussian distribution, policy improvement theorem.

 %Reinforcement learning, exploration, exploitation,
 %entropy regularization, stochastic control, relaxed control, linear--quadratic, Gaussian distribution.
\end{abstract}
\newpage

\section{Introduction}

Applications of reinforcement learning (RL) to quantitative finance (e.g., algorithmic and high frequency trading, smart order routing, portfolio management, etc) have attracted  more attentions in recent years. One of the main reasons is that the electronic markets prevailing nowadays can provide sufficient amount of microstructure data for training and adaptive learning, much beyond what human traders and portfolio managers could handle in old days. Numerous studies  have been carried out along  this direction. For example, \cite{NFK} conducted the first large scale empirical analysis of RL method applied to optimal order execution and achieved substantial improvement relative to the baseline strategies. \cite{HW} improved over the theoretical optimal trading strategies of the Almgren-Chriss model (\cite{AC}) using RL techniques and market attributes. \cite{MS} and \cite{MS2} studied portfolio allocation problems with transaction costs via direct policy search based RL methods, without resorting to forecast models that rely on supervised leaning.

However, most existing works only focus on RL optimization problems with expected utility of discounted rewards. Such criteria are either unable to fully characterize the uncertainty of the decision making process in financial  markets or opaque to  typical investors. On the other hand,
mean--variance (MV)  is one of the most important criteria for portfolio choice. Initiated in the seminal work \cite{M} for portfolio selection  in a single period, such a criterion yields an asset allocation strategy that minimizes the variance of the final payoff while targeting some prespecified mean return. The MV problem has been further investigated in the discrete-time multiperiod setting  (\cite{LN}) and the continuous-time setting (\cite{MV_zhou}), along with hedging  (\cite{DR}) and optimal liquidation  (\cite{AC}), among many other variants and generalizations. The popularity of the MV criterion is not only due to its intuitive and transparent nature  in capturing the tradeoff between  risk and reward for practitioners, but also due to the theoretically interesting  issue of time-inconsistency (or Bellman's inconsistency) inherent with the underlying stochastic optimization and control problems.

From the RL perspective, it is computationally challenging to  seek the global optimum for Markov Decision Process (MDP) problems under the MV criterion (\cite{MT}).
In fact, variance estimation and control are not as direct as optimizing the expected reward-to-go which has been well understood in the classical MDP framework for most RL problems. Because most standard MDP performance criteria  are linear in expectation, including the discounted sum of rewards and the long-run average reward (\cite{SB}), Bellman's consistency equation can be easily derived for guiding policy evaluation and control, leading  to many state-of-the-art RL techniques (e.g., Q-learning, temproaral difference (TD) learning, etc). The variance of reward-to-go, however, is nonlinear in expectation and, as a result, most of the well-known learning rules cannot be applied directly.

Existing works on variance estimation and control generally divide into two groups, value based methods and policy based methods.  \cite{Sobel} obtained the Bellman's equation for the variance of reward-to-go under a {\it fixed}, given policy. Based on that equation,  \cite{SKK} derived the TD(0) learning rule to estimate the variance under any given policy. In a related paper,  \cite{SK} applied  this value based method  to an MV  portfolio selection problem. It is worth noting that due to their definition of the intermediate value function (i.e., the variance penalized expected reward-to-go),  Bellman's optimality principle does not hold. As a result, it is not guaranteed that a greedy policy based on the latest updated value function will eventually lead to the true global optimal policy. The second approach, the policy based RL,  was proposed  in \cite{TDM}. They also extended the work to linear function approximators and devised actor-critic algorithms for  MV optimization problems for which convergence to the local optimum is guaranteed with probability one (\cite{TM}). Related works following this line of research include \cite{MV_2,MV_1}, among others. Despite the various methods mentioned above, it remains  an open and interesting question in RL to search for the {\it global} optimum under the MV criterion.

In this paper, we establish an RL framework for studying  the continuous-time MV portfolio selection, with
continuous portfolio (control/action) and wealth (state/feature) spaces. The continuous-time formulation is  appealing when the rebalancing of portfolios can take place  at ultra-high frequency. Such a formulation may also  benefit from the large amount of tick data that is available  in most electronic markets nowadays. The classical continuous-time MV portfolio
model is a spacial instance of a stochastic linear--quadratic (LQ) control problem
(\cite{MV_zhou}). Recently, \cite{Hwang} proposed and developed a general entropy-regularized,  relaxed stochastic control formulation, called an {\it exploratory} formulation,  to capture explicitly
the tradeoff between exploration and exploitation in RL. They showed that the optimal distributions of the exploratory control policies must be Gaussian for an LQ control problem in the infinite time horizon, thereby providing an interpretation for the Guassian exploration broadly used both in RL algorithm design and in practice.

While being essentially an  LQ control problem, the MV portfolio selection must be formulated in a {\it finite} time horizon which is not covered by \cite{Hwang}.
The first contribution of this paper is to present the global optimal solution to the exploratory MV problem. One of the interesting findings is that, unlike its infinite horizon counterpart derived in \cite{Hwang}, the optimal feedback control policy for the finite horizon case is a Gaussian distribution with a {\it time-decaying} variance.
This suggests that the level of exploration decreases as the time approaches the
end of the planning horizon.
On the other hand, we will obtain results and observe insights that are parallel  to those in \cite{Hwang}, such as the perfect separation between exploitation and exploration in the mean and variance of the optimal Gaussian distribution, the positive effect of a random environment on learning, and the  close connections between the classical and the exploratory MV problems.

The main contribution of the paper, however, is to
design an interpretable and implementable RL algorithm to learn the global optimal solution of the exploratory MV problem, premised upon a provable {\it policy improvement theorem} for continuous-time stochastic control problems with both entropy regularization and control relaxation. This theorem provides an explicit updating scheme for the feedback Gaussian policy, based on the value function of the current policy in an iterative manner. Moreover, it enables us to
reduce  from a family of general non-parametric policies to a specifically parametrized Gaussian family for exploration and exploitation, irrespective of the choice of an initial policy. This, together with a carefully chosen initial Gaussian policy at the beginning of the learning process, guarantees the fast convergence of both the policy and the value function to the global optimum of the exploratory MV problem.

We further compare our RL algorithm with two other methods applied to the MV portfolio optimization. The first one is an adaptive control approach that adopts the real-time maximum likelihood estimation of the underlying model parameters.
The other one is a recently developed continuous control RL algorithm, a deep deterministic policy gradient method (\cite{Lillicrap}) that employs deep neural networks. The comparisons are performed under various simulated market scenarios, including those with both stationary and non-stationary investment opportunities. In nearly all the simulations, our RL algorithm outperforms the other two methods by a large margin, in terms of both performance and training time.

The rest of the paper is organized as follows. In Section 2, we present the continuous-time exploratory MV problem under the entropy-regularized relaxed stochastic control framework. Section 3 provides the complete solution of the exploratory MV problem, along with  connections to its classical counterpart. We then provide the policy improvement theorem and a convergence result for the learning problem in Section 4, based on which we devise the RL algorithm for solving the exploratory MV problem. In Section 5, we compare our algorithm with two other methods in  simulations under various market scenarios. Finally, we conclude in Section 6.

\section{Formulation of Problem}

In this section, we formulate an exploratory, entropy-regularized Markowitz's MV portfolio selection problem in continuous time, in the context of RL. The motivation of a general exploratory stochastic control formulation, of which the MV problem is a special case,  was discussed at great length  in a previous paper  \cite{Hwang}; so we will frequently refer to that paper.

\subsection{Classical continuous-time MV problem}

We first recall the classical MV problem in continuous time (without RL).
For ease of presentation, throughout this paper we consider an investment universe consisting of only one risky asset and one riskless asset. The case of multiple risky assets
poses no essential differences or difficulties other than notational complexity.

Let an investment  planning horizon $T>0$ be fixed, and $\{W_t,0\leq t\leq T\}$ a standard one-dimensional Brownian motion defined on a filtered probability space $(\Omega, \mathcal{F},\{\mathcal{F}_{t}\}_{0\leq t\leq T},\mathbb{P})$ that satisfies the usual conditions. The price process of the risky asset is a geometric Brownian motion governed by
\begin{equation}\label{price}
dS_t=S_t\left(\mu\, dt+\sigma \,dW_t\right), \quad 0\leq t\leq T,
\end{equation}
with $S_0=s_0>0$ being the initial price at $t=0$, and $\mu\in \mathbb{R}$, $\sigma>0$ being the mean and volatility parameters, respectively. The riskless asset
has  a constant interest rate $r>0$. \footnote{In practice, the true  (yet unknown) investment opportunity parameters $\mu$, $\sigma$ and $r$ can be time-varying stochastic processes. Most existing quantitative finance methods are devoted to
estimating these parameters. In contrast, RL learns the values of various strategies and the optimal value through exploration and exploitation, {\it without} assuming any statistical properties of these parameters or estimating them.
But for a model-based, classical  MV problem, we assume these parameters are constant and known. In subsequent contexts, all we need is the structure of the problem for our RL algorithm design.}
%
% that move at much slower paces  than the learning process.
%
%
%Our proposed RL algorithm can still work in such non-stationary environment; see, for example, section 3.2 in \cite{kushner}.} $r>0$ for the riskless asset over the investment horizon $[0,T]$.
The Sharpe ratio of the risky asset is defined by $\rho=\frac{\mu-r}{\sigma}$.

Denote by $\{x^u_t,0\leq t\leq T\}$ the {\it discounted} wealth process of an agent who rebalances her portfolio investing in  the risky and riskless assets with a strategy $u=\{u_t, 0\leq t\leq T\}$. Here $u_t$ is the discounted dollar value put  in the risky asset at time $t$, while satisfying the standard self-financing assumption and other technical conditions that will be spelled  out  in details below. It follows  from (\ref{price}) that the wealth process satisfies
\begin{equation}\label{classical_wealth}
dx^u_t=\sigma u_t(\rho\, dt+\,dW_t), \quad 0\leq t\leq T,
\end{equation}
with an initial endowment being $x^u_0=x_0\in \mathbb{R}$.

The classical continuous-time MV model aims to solve the following constrained optimization problem
\begin{eqnarray}\nonumber
& &\min_{u} \text{Var}[x^u_T]\\
& &\text{subject to} \ \mathbb{E}[x^u_T]=z,\label{target}
\end{eqnarray}
where $\{x^u_t,0\leq t\leq T\}$ satisfies the dynamics (\ref{classical_wealth}) under the investment strategy (portfolio) $u$, and $z\in \mathbb{R}$ is an investment target set at $t=0$ as  the desired mean payoff at the end of the investment horizon $[0,T]$.\footnote{The original  MV problem is to find the Pareto efficient frontier for a two-objective (i.e. maximizing the expected terminal payoff and minimizing its variance) optimization problem. There are a number of equivalent mathematical formulations to find such a frontier, (\ref{target}) being one of them. In particular, by varying the parameter $z$ one can trace out the frontier. See \cite{MV_zhou} for details.}

Due to the variance in its objective, (\ref{target}) is known to be {\it time inconsistent}. The problem then becomes descriptive rather than normative because there is generally no {\it dynamically} optimal solution for a time-inconsistent optimization problem. Agents react differently to the same time-inconsistent problem, and a goal of the study becomes to {\it describe} the different behaviors when facing
such time-inconsistency.
In this paper we focus ourselves to the so-called \textit{pre-committed} strategies of the MV problem, which are optimal at $t=0$ only.\footnote{For a detailed discussions about the different behaviors
under time-inconsistency, see the seminal paper \cite{strotz}. Most of the study
on continuous-time MV problem in literature has been devoted to pre-committed strategies; see \cite{MV_zhou, MV_no_shorting, MV_no_bankruptcy, MV_random, MV_regime_switching}.}

To solve (\ref{target}), one first transforms it into an unconstrained problem by applying a Lagrange multiplier $w$:\footnote{Strictly speaking, $2w\in \mathbb{R}$ is the Lagrange multiplier.}
\begin{equation}\label{unconstrained_classical}
\min_u \mathbb{E}[(x^u_T)^2]-z^2-2w\left(\mathbb{E}[x^u_T]-z\right)=\min_u \mathbb{E}[(x^u_T-w)^2]-(w-z)^2.
\end{equation}
This problem can be solved analytically, whose solution $u^*=\{u^*_t,0\leq t\leq T\}$ depends on $w$.
Then the original constraint $\mathbb{E}[x^{u^*}_T]=z$ determines the value of
$w$.  We refer a detailed derivation to \cite{MV_zhou}.

\subsection{Exploratory continuous-time MV problem}
Given the complete knowledge of the model parameters, the classical, model-based
 MV problem (\ref{target}) and many of its variants have been solved rather
 completely. When implementing these solutions, one needs to estimate
 the market parameters from historical time series of asset prices, a procedure known as {\it identification} in classical adaptive control. However, it is well known that in practice  it is difficult to estimate the investment opportunity parameters, especially the mean return (aka the {\it mean--blur problem}; see, e.g.,  \cite{Luenberger}) with a workable  accuracy. Moreover, the classical optimal MV strategies are often extremely sensitive to these parameters,  largely due to the procedure of inverting ill-conditioned variance--covariance matrices  to obtain optimal allocation weights. In view of these two issues, the Markowitz solution can be greatly irrelevant to the underlying investment objective.

On the other hand, RL techniques do not require, and indeed often skip, any estimation of model parameters. Rather, RL algorithms, driven by historical data,  output optimal (or near-optimal) allocations directly. This is made possible by
direct interactions with the unknown investment environment, in a learning (exploring) while optimizing (exploiting) fashion.
\cite{Hwang} motivated and proposed a general theoretical framework for exploratory, RL stochastic control problems and carried out a detailed study for the special LQ case, albeit in the setting of the infinite time horizon.
We adopt the same framework here, noting the inherent features of an LQ structure and a {\it finite} time horizon of the MV problem. Indeed, although the motivation for  the exploratory formulation is mostly the same, there are intriguing new insights emerging with this transition from the infinite time horizon to its finite counterpart.

First, we introduce the ``exploratory" version of the state dynamics (\ref{classical_wealth}). It was  originally proposed  in \cite{Hwang}, motivated by
repetitive learning in RL. In this formulation, the control (portolio) process $u=\{u_t, 0\leq t\leq T\}$ is randomized, which represents exploration and learning, leading to
a measure-valued or distributional control process whose density function is given by $\pi=\{\pi_t,0\leq t\leq T\}$. The dynamics (\ref{classical_wealth}) is changed to
\begin{eqnarray}
dX^{\pi}_t &=&\tilde{b}(\pi_t)\,dt+\tilde{\sigma}(\pi_t)\,dW_t,
 \label{state_process0}
\end{eqnarray}
where $0<t\leq T$ and $X^{\pi}_0=x_0$,
\begin{equation}\label{drift}
\tilde{b}(\pi):=\int_{\mathbb{R}} \rho \sigma u\pi(u)du,\;\;\pi\in \mathcal{P}\left( \mathbb{R}\right),
\end{equation}
and
\begin{equation}\label{volatility}
\tilde{\sigma}(\pi):=\sqrt{\int_{\mathbb{R}} \sigma^2 u^2\pi(u)du},\;\;\pi\in \mathcal{P}\left( \mathbb{R}\right),
\end{equation}
with $\mathcal{P}%
\left( \mathbb{R}\right) $ being the set of density functions of probability measures on $\mathbb{R}$ that are absolutely
continuous with respect to the Lebesgue measure. Mathematically,
(\ref{state_process0}) coincides with the {\it relaxed control} formulation in classical control theory.
Refer to \cite{Hwang} for a  detailed discussion on the motivation of (\ref{state_process0}).

Denote respectively by $\mu_t$ and $\sigma^2_t$, $0\leq t\leq T$, the mean and variance (assuming they exist for now) processes associated with the distributional control process $\pi$, i.e.,
\begin{equation}\label{mean-variance_process}
\mu _{t}:=\int_{\mathbb{R}}u\pi _{t}(u)du\quad \  \  \text{and}\quad \ \ \sigma
_{t}^{2}:=\int_{\mathbb{R}}u^{2}\pi _{t}(u)du-\mu _{t}^{2}\text{ }.
\end{equation}%
Then, it follows immediately that the exploratory dynamics (\ref{state_process0}) become
\begin{eqnarray}
dX^{\pi}_t
&=&\rho \sigma \mu_t\, dt+\sigma \sqrt{\mu_t^2+\sigma_t^2}\, dW_t, \label{state_process}
\end{eqnarray}
where $0<t\leq T$ and $X^{\pi}_0=x_0$.
The randomized, distributional  control process $\pi=\{\pi_t, 0\leq t\leq T\}$ is to model
exploration, whose overall level is in turn captured by
its accumulative differential entropy
\begin{equation}\label{entropy}
\mathcal{H}(\pi):=-\int_0^T\int_{\mathbb{R}}\pi_t(u)\ln \pi_t(u)dudt.
\end{equation}
Further, introduce a {\it temperature parameter} (or {\it exploration weight}) $\lambda>0$ reflecting  the tradeoff between exploitation and exploration.
The entropy-regularized, exploratory  MV problem is then to solve, for any fixed $w\in \mathbb{R}$:
\begin{equation}\label{value_function}
\min_{\pi\in \mathcal{A}(x_0,0)}\mathbb{E}\left[(X_T^{\pi}-w)^2+\lambda \int_0^T\int_{\mathbb{R}}\pi_t(u)\ln \pi_t(u)dudt\right]-(w-z)^2,
\end{equation}
where $\mathcal{A}(x_0,0)$ is the set of admissible distributional controls on $[0,T]$ to be precisely defined below. Once this problem is solved with a minimizer
$\pi^*=\{\pi^*_t, 0\leq t\leq T\}$, the Lagrange multiplier $w$ can  be determined by the additional constraint $\mathbb{E}[X^{\pi^*}_T]=z$.

The optimization objective (\ref{value_function}) explicitly encourages exploration, in contrast to the classical problem (\ref{unconstrained_classical}) which concerns exploitation only.

We will solve (\ref{value_function}) by dynamic programming. For that we need to define the value functions. For each $(s,y)\in[0,T)\times \mathbb{R}$, consider the
state equation (\ref{state_process}) on $[s,T]$ with $X^{\pi}_s=y$.
Define the set of admissible controls, $\mathcal{A}(s,y)$, as follows. Let $\mathcal{B}(\mathbb{R})$ be the Borel algebra on $\mathbb{R}$. A (distributional) control (or portfolio/strategy) process $\pi=\{\pi_t,s\leq t\leq T\}$ belongs to $\mathcal{A}(s,y)$, if

\smallskip

(i)		for each $s\leq t\leq T$, $\pi _{t}\in \mathcal{P}(U)$ a.s.;

(ii)		for each $A\in \mathcal{B}(\mathbb{R})$, $\{\int_A\pi _{t}(u)du,s\leq t\leq T\} $
is $\mathcal{F}_{t}$-progressively measurable;

(iii)	$\mathbb{E}\left[
\int_{s}^{T}\left( \mu _{t}^{2}+\sigma _{t}^{2}\right) dt\right] <\infty$;

(iv)	$\mathbb{E}\left[\big|(X_T^{\pi}-w)^2+\lambda \int_s^T\int_{\mathbb{R}}\pi_t(u)\ln \pi_t(u)dudt\big|\; \Big | X_s^{\pi}=y\right]<\infty$.

\medskip

Clearly, it follows from condition (iii) that
the stochastic differential equation (SDE) (\ref{state_process}) has a unique strong solution for $s\leq t\leq T$ that satisfies $X^{\pi}_s=y$.
%\footnote{Actually (\ref{state_process}) is a trivial SDE in that its right hand side does not depend on the unknown $X^\pi$.}

\medskip

Controls in $\mathcal{A}(s,y)$ are measure-valued (or, precisely, density-function-valued) stochastic {\it processes}, which are also called {\it open-loop} controls in the control terminology.
As in the classical control theory, it is important to distinguish between open-loop controls and {\it feedback} (or {\it closed-loop}) controls (or {\it policies} as in the RL literature, or {\it laws} as in the control literature). Specifically, a {\it deterministic} mapping $\boldsymbol{\pi}(\cdot;\cdot,\cdot)$ is called an (admissible)  feedback control  if i) $\boldsymbol{\pi}(\cdot;t,x)$ is a density function for each $(t,x)\in[0,T]\times \mathbb{R}$; ii) for each $(s,y)\in[0,T)\times \mathbb{R}$, the following SDE (which is the system dynamics after the feedback policy $\boldsymbol{\pi}(\cdot;\cdot,\cdot)$ is applied)
\begin{equation}\label{new_dynamics_feedback}
dX^{\boldsymbol{\pi}}_t=\tilde{b}( \boldsymbol{\pi}(\cdot;t,X^{\boldsymbol{\pi}}_t))dt+\tilde{\sigma}(\boldsymbol{\pi}(\cdot;t,X^{\boldsymbol{\pi}}_t))dW_t,\; t\in[s,T]; \ X^{\boldsymbol{\pi}}_{s}=y
\end{equation}
has a unique strong solution $\{X^{\boldsymbol{\pi}}_t,t\in[s,T]\}$,  and  the open-loop control
$\pi=\{\pi
_{t},$ $t\in[s,T]\}\in \mathcal{A}(s,y)$ where $\pi_{t}:=\boldsymbol{\pi}(\cdot;t,X_t^{\boldsymbol{\pi}})$. In this case, the open-loop control $\pi$ is said to be
{\it generated} from the feedback policy $\boldsymbol{\pi}(\cdot;\cdot,\cdot)$ {\it with respect to} the initial time and state,  $(s,y)$. It is useful to note that an open-loop control and its admissibility depend on the initial $(s,y)$, whereas a
feedback policy can generate open-loop controls for {\it any} $(s,y)\in[0,T)\times \mathbb{R}$, and hence is in itself independent of $(s,y)$.\footnote{Throughout this paper, we use boldfaced $\boldsymbol{\pi}$ to denote feedback controls, and the normal style $\pi$ to denote open-loop controls.}

Now, for a fixed $w\in \mathbb{R}$, define
\begin{equation}\label{value_function_general}
V(s,y;w):=\inf_{\pi\in \mathcal{A}(s,y)}\mathbb{E}\left[(X_T^{\pi}-w)^2+\lambda \int_0^T\int_{\mathbb{R}}\pi_t(u)\ln \pi_t(u)dudt\Big | X_s^{\pi}=y\right]-(w-z)^2,
\end{equation}
for $(s,y)\in[0,T)\times \mathbb{R}$.
The function $V(\cdot,\cdot;w)$ is called the {\it optimal  value function} of the problem.\footnote{In the control literature, $V$ is called the value function.
However, in the RL literature the term ``value function" is also used for the objective value under a particular control. So to avoid ambiguity we call $V$ the {\it optimal} value function.}
Moreover, we define  the {\it value function}  under any given {\it feedback} control  $\boldsymbol{\pi}$:

\begin{equation}\label{general_value}
V^{\boldsymbol{\pi}}(s,y ;w)=\mathbb{E}\left[(X_T^{\boldsymbol{\pi}}-w)^2+\lambda \int_s^T\int_{\mathbb{R}}\pi_t(u)\ln \pi_t(u)dudt\Big | X_s^{\boldsymbol{\pi}}=y\right]-(w-z)^2,
\end{equation}
for $(s,y)\in[0,T)\times \mathbb{R}$, where $\pi=\{\pi
_{t},$ $t\in[s,T]\}$ is  the open-loop control generated from $\boldsymbol{\pi}$  with respect to  $(s,y)$ and $\{X^{\boldsymbol{\pi}}_t,t\in[s,T]\}$ is the corresponding wealth process.

%Finally, we define the set of admissible controls, $\mathcal{A}(x,0)$, as follows. Let $\mathcal{B}(\mathbb{R})$ be the Borel algebra on $\mathbb{R}$. An (distributional) control $\pi=\{\pi_t,0\leq t\leq T\}$, with $\pi_t$ being a density function, belongs to $\mathcal{A}(x,0)$, if
%
%\smallskip
%
%(i)		for each $0\leq t\leq T$, $\pi _{t}\in \mathcal{P}(U)$ a.s.;
%
%(ii)		for each $A\in \mathcal{B}(\mathbb{R})$, $\{\int_A\pi _{t}(u)du,0\leq t\leq T\} $
%is $\mathcal{F}_{t}$-progressively measurable;
%
%(iii)	$\mathbb{E}\left[
%\int_{0}^{T}\left( \mu _{s}^{2}+\sigma _{s}^{2}\right) ds\right] <\infty$;
%
%(iv)	$\mathbb{E}\left[\big|(X_T^{\pi}-w)^2+\lambda \int_0^T\int_{\mathbb{R}}\pi_t(u)\ln \pi_t(u)dudt\big|\; \Big | X_0^{\pi}=x\right]<\infty$.
%
%
%Clearly, it follows from condition (iii) that the stochastic differential equation (SDE) (\ref{state_process}) has a unique strong solution for $0\leq t\leq T$ that satisfies $X^{\pi}_0=x$.

\section{Solving Exploratory MV Problem}

In this section we first solve the exploratory MV problem, and then establish solvability equivalence  between the classical and exploratory problems. The latter is important for understanding the cost of exploration and for devising
RL algorithms.

\subsection{Optimality of Gaussian exploration}

To solve the exploratory MV problem (\ref{value_function}), we apply the classical Bellman's principle of optimality:
$$V(t,x;w)=\inf_{\pi\in \mathcal{A}(t,x)}\mathbb{E}\left[V(s,X^{\pi}_s;w)+\lambda\int_t^s\int_{\mathbb{R}}\pi_v(u)\ln \pi_v(u)dudv\Big | X^{\pi}_t=x\right],$$
for $x\in \mathbb{R}$ and $0\leq t< s\leq T$. Following standard arguments, we deduce that $V$ satisfies the Hamilton-Jacobi-Bellman (HJB) equation
\begin{equation}\label{HJB_1}
v_t(t,x;w)+\min_{\pi\in \mathcal{P}(\mathbb{R})}\Big(\frac{1}{2}\tilde{\sigma}^2(\pi)v_{xx}(t,x;w)+\tilde{b}(\pi)v_x(t,x;w)+\lambda \int_{\mathbb{R}}\pi(u)\ln \pi(u)du\Big)=0,
\end{equation}
or, equivalently,
\begin{equation}\label{HJB_2}
v_t(t,x;w)+\min_{\pi\in \mathcal{P}(\mathbb{R})}\int_{\mathbb{R}}\left(\frac{1}{2}\sigma^2 u^2 v_{xx}(t,x;w)+\rho\sigma u v_x(t,x;w)+\lambda \ln\pi(u)\right)\pi(u)du=0,
\end{equation}
with the terminal condition $v(T,x;w)=(x-w)^2-(w-z)^2$. Here $v$ denotes the generic unknown solution to the HJB equation.

Applying the usual verification technique and using the fact that $\pi\in \mathcal{P}(\mathbb{R})$ if and only if
\begin{equation}\label{constrained_problem}
\int_{\mathbb{R}}\pi (u)du=1\quad \text{and}\quad \pi (u)\geq
0\  \text{a.e.}\quad  \text{on}\ \mathbb{R},
\end{equation}
we can solve the (constrained) optimization problem in the HJB equation (\ref{HJB_2}) to obtain a feedback  (distributional) control whose density function is given by
\begin{eqnarray}\nonumber
\boldsymbol{\pi} ^{\ast }(u;t,x,w)&=&\frac{\exp\left(-\frac{1}{\lambda}\left(\frac{1}{2}\sigma^2 u^2 v_{xx}(t,x;w)+\rho \sigma v_x(t,x;w)\right)\right)}{\int_{\mathbb{R}}\exp\left(-\frac{1}{\lambda}\left(\frac{1}{2}\sigma^2 u^2 v_{xx}(t,x;w)+\rho \sigma v_x(t,x;w)\right)\right)du}\\
&=& \mathcal{N}\left(\, u\, \Big | -\frac{\rho}{\sigma}\frac{v_x(t,x)}{v_{xx}(t,x;w)}\ , \ \frac{\lambda}{\sigma^2 v_{xx}(t,x;w)}\right),\label{Gaussian}
\end{eqnarray}
where we have denoted by $\mathcal{N}(u|\alpha,\beta)$ the Gaussian density function with mean $\alpha\in \mathbb{R}$ and variance $\beta>0$. In the above representation, we have assumed that $v_{xx}(t,x;w)>0$, which will be verified in what follows.

Substituting the candidate optimal Gaussian feedback control policy (\ref{Gaussian}) back into the HJB equation (\ref{HJB_2}), the latter is transformed to
\begin{equation}\label{HJB_3}
v_t(t,x;w)-\frac{\rho^2}{2}\frac{v^2_x(t,x;w)}{v_{xx}(t,x,w)}+\frac{\lambda}{2}\left(1-\ln \frac{2\pi e \lambda}{\sigma^2v_{xx}(t,x;w)}\right)=0,
\end{equation}
with $v(T,x;w)=(x-w)^2-(w-z)^2$. A direct computation yields that this equation has a classical solution
\begin{equation}\label{solution_to_HJB}
v(t,x;w)=(x-w)^2e^{-\rho^2(T-t)}+\frac{\lambda \rho ^2}{4}\left(T^2-t^2\right)-\frac{\lambda}{2}\left(\rho^2T-\ln \frac{\sigma^2}{\pi \lambda}\right)(T-t)-(w-z)^2,
\end{equation}
which clearly satisfies  $v_{xx}(t,x;w)>0$, for any $(t,x)\in [0,T]\times \mathbb{R}$. It then follows that the candidate optimal feedback Gaussian control (\ref{Gaussian}) reduces to
\begin{equation}\label{Gaussian_explicit}
\boldsymbol{\pi} ^{\ast }(u;t,x,w)=\mathcal{N}\left(\, u\, \Big| -\frac{\rho}{\sigma}(x-w)\ , \  \frac{\lambda}{2\sigma^2}e^{\rho^2(T-t)}\right),\;\;
(t,x)\in [0,T]\times \mathbb{R}.
\end{equation}
%for $(t,x)\in [0,T]\times \mathbb{R}$.

Finally, the  optimal wealth process (\ref{state_process}) under $\boldsymbol{\pi} ^{\ast }$ becomes
\begin{equation}\label{wealth_SDE}
dX^{*}_t=-\rho^2(X^{*}_t-w)\, dt+\sqrt{\rho^2\left(X^{*}_t-w\right)^2+\frac{\lambda}{2}e^{\rho^2(T-t)}}\, dW_t,\;
X^{*}_0=x_0.
\end{equation}
It has a unique strong solution for $0\leq t\leq T$, as can be easily verified.

We now summarize the above results in the following theorem.

\begin{theorem}\label{verification}
The optimal value function of the entropy-regularized exploratory MV problem (\ref{value_function}) is given by
\begin{equation}\label{value_verified}
V(t,x;w)=(x-w)^2e^{-\rho^2(T-t)}+\frac{\lambda \rho ^2}{4}\left(T^2-t^2\right)-\frac{\lambda}{2}\left(\rho^2T-\ln \frac{\sigma^2}{\pi \lambda}\right)(T-t)-(w-z)^2,
\end{equation}%
for $(t,x)\in [0,T]\times \mathbb{R}$.
Moreover,
the optimal feedback control is Gaussian, with its density function given by
\begin{equation}
\boldsymbol{\pi}^{\ast }(u;t,x,w)=\mathcal{N}\left( \, u\, \Big| -\frac{\rho}{\sigma}(x-w)\ , \  \frac{\lambda}{2\sigma^2}e^{\rho^2(T-t)} \right).
\label{Gaussian_verified}
\end{equation}%
The associated optimal wealth process under $\boldsymbol{\pi}^{\ast }$
is the unique solution of the SDE
\begin{equation}\label{wealth_SDE_1}
dX^{*}_t=-\rho^2(X^{*}_t-w)\, dt+\sqrt{\rho^2\left(X^{*}_t-w\right)^2+\frac{\lambda}{2}e^{\rho^2(T-t)}}\, dW_t,\;X^{*}_0=x_0.
\end{equation}
Finally, the Lagrange multiplier $w$ is given by $w=\frac{ze^{\rho^2T}-x_0}{e^{\rho^2T}-1}$.
\end{theorem}
\begin{proof}
For each fixed $w\in \mathbb{R}$, the verification arguments aim to show that the optimal value function of problem (\ref{value_function}) is given by (\ref{value_verified}) and that the candidate optimal policy (\ref{Gaussian_verified}) is indeed admissible. A detailed proof follows the same lines of that of Theorem $4$ in \cite{Hwang}, and is left for interested readers.

We now determine the Lagrange multiplier $w$ through the constraint $\mathbb{E}[X^*_T]=z$. It follows  from (\ref{wealth_SDE_1}),
% that
%$$X^*_t=x+\int_0^t- \rho^2(X^{*}_s-w)\, ds+\int_0^t\sqrt{\rho^2\left(X^{*}_s-w\right)^2+\frac{\lambda}{2}e^{\rho^2(T-s)}}\, dW_s,$$
%for $0\leq t\leq T$. Using
along with the standard estimate that $\mathbb{E}\left[\max_{t\in [0,T]}(X^*_t)^2\right]<\infty$ and Fubini's Theorem, that
$$\mathbb{E}[X^*_t]=x_0+\mathbb{E}\left[\int_0^t -\rho^2(X^*_s-w)\, ds\right]=x_0+\int_0^t -\rho^2\left(\mathbb{E}[X^*_s]-w\right)\, ds.$$
Hence,  $\mathbb{E}[X^*_t]=(x_0-w)e^{-\rho^2 t}+w$. The constraint $\mathbb{E}[X^*_T]=z$ now becomes $(x_0-w)e^{-\rho^2T}+w=z$, which gives  $w=\frac{ze^{\rho^2T}-x_0}{e^{\rho^2T}-1}$.
\end{proof}

\smallskip

There are several interesting points to note in this result.
First of all, it follows from Theorem \ref{Theorem_equivalence} in the next section that the classical and the exploratory MV problems have the {\it same} Lagrange multiplier
value due to the fact that the optimal terminal wealths under the respective optimal
feedback controls of the two  problems turn out to have the same mean.\footnote{Theorem \ref{Theorem_equivalence} is a reproduction of the results on the classical MV problem obtained in \cite{MV_zhou}.}
 This latter result is rather surprising at first sight because the exploration greatly alters the underlying system dynamics (compare the dynamics
 (\ref{classical_wealth}) with (\ref{state_process})).

Second, the variance of the optimal Gaussian policy, which measures the level of exploration,   is $\frac{\lambda}{2\sigma^2}e^{\rho^2(T-t)}$ at time $t$. So the exploration decays in time: the agent initially engages in exploration at the maximum level, and reduces it gradually (although never to zero) as time passes and approaches the end of the investment horizon.
Hence, different from its  infinite horizon counterpart  studied in \cite{Hwang}, the extent of exploration  is no longer constant, but, rather, annealing. This is intuitive because, %in the initial stage, the agent knows little about the market so she needs more exploration. As
as the RL agent learns more about the random environment as time passes, the exploitation becomes more important since there is a deadline $T$ at which her actions will be evaluated. Naturally, exploitation dominates exploration as time approaches maturity.  Theorem \ref{verification} presents such a  decaying exploration scheme {\it endogenously} which, to our best knowledge, has not been derived in the RL literature.

Third, as already noted in \cite{Hwang}, at any given $t\in [0,T]$, the variance of the exploratory Gaussian distribution decreases as the volatility of the risky asset increases, with other parameters being fixed.
The volatility of the risky asset reflects the level of randomness of the investment universe. This hints that a more random environment contains more learning
opportunities, which the RL agent can leverage to reduce her own exploratory endeavor
because, after all, exploration is costly.

Finally, the mean of the Gaussian distribution (\ref{Gaussian_verified}) %coincides with the true optimal allocation of the classical Mean-variance problem (cf. (\ref{classical_optimal_control})), and
is independent of the exploration weight $\lambda$, while its variance is independent of the state $x$. This highlights  a \textit{perfect separation} between exploitation and exploration, as the former is captured by the mean and  the latter  by the variance of the optimal Gaussian exploration. This property is also consistent with the LQ case in the  infinite horizon  studied in \cite{Hwang}.

\subsection{Solvability equivalence between classical and exploratory MV problems}
In this section, we establish the solvability equivalence between the classical and the exploratory, entropy-regularized MV problems. Note that both problems
can be and indeed have been solved separately and independently. Here by ``solvability equivalence" we mean that the solution of one problem will lead to that of the other
{\it directly}, without needing to solve it separately. This equivalence was first discovered in \cite{Hwang} for the infinite horizon LQ case, and was shown to be instrumental in
deriving  the convergence result (when the exploration weight $\lambda$ decays to $0$)  as well as analyzing the exploration cost therein. Here, the discussions are  mostly parallel; so they will be brief.

Recall the classical MV problem (\ref{unconstrained_classical}). In order to apply dynamic programming, we again consider the set of admissible controls, $\mathcal{A}^{\text{cl}}(s,y)$,
for $(s,y)\in [0,T)\times \mathbb{R}$,
\begin{center}
$\mathcal{A}^{\text{cl}}(s,y):=\Big\{u=\{u_t, t\in [s,T]\}$: $u$ is $\mathcal{F}_t$-progressively measurable and $\mathbb{E}[\int_s^T(u_s)^2\, ds]<\infty\Big\}.$
\end{center}
The (optimal) value function is defined by
\begin{equation}\label{classical_value_function}
V^{\text{cl}}(s,y;w):=\inf_{u\in \mathcal{A}^{\text{cl}}(s,y)}\mathbb{E}\big [(x^u_T-w)^2\, \big | \, x^u_s=y\big ]-(w-z)^2,
\end{equation}
for $(s,y)\in [0,T)\times \mathbb{R}$, where $w\in \mathbb{R}$ is fixed. Once this problem is solved, $w$ can be determined by the constraint $\mathbb{E}[x^*_T]=z$, with $\{x^*_t, t\in [0,T]\}$ being the optimal wealth process under the optimal portfolio $u^*$.

The HJB equation is
\begin{equation}\label{HJB_classical}
\omega_t(t,x;w)+\min_{u\in \mathbb{R}}\left(\frac{1}{2}\sigma^2 u^2\omega_{xx}(t,x;w)+\rho \sigma u \, \omega_x(t,x;w)\right)=0,\;\;(t,x)\in [0,T)\times \mathbb{R},
\end{equation}
with the terminal condition  $\omega(T,x;w)=(x-w)^2-(w-z)^2$.

Standard verification arguments  deduce the optimal value function to be
$$V^{\text{cl}}(t,x;w)=(x-w)^2e^{-\rho^2(T-t)}-(w-z)^2,$$
the optimal feedback control policy to be
\begin{equation}\label{classical_optimal_control}
\boldsymbol{u}^{\ast }(u;t,x,w)=-\frac{\rho}{\sigma}(x-w),
\end{equation}
and the corresponding optimal wealth process to be the unique strong solution to the SDE
\begin{equation}\label{classical_optimal_wealth}
dx^*_t=-\rho^2(x^*_t-w)\, dt-\rho (x^*_t-w)\, dW_t,\quad x^*_0=x_0.
\end{equation}
Comparing the optimal wealth dynamics, (\ref{wealth_SDE_1}) and (\ref{classical_optimal_wealth}), of the exploratory and classical problems,
we note that they have the same {\it drift} coefficient (but different {\it diffusion} coefficients). As a result, the two problems have the same
mean of optimal terminal wealth and hence the same value of
the Lagrange multiplier $w=\frac{ze^{\rho^2T}-x_0}{e^{\rho^2T}-1}$ determined by the constraint $\mathbb{E}[x^*_T]=z$.

We now provide the solvability equivalence between the two problems. The proof is very similar to that of Theorem $7$ in \cite{Hwang}, and is thus omitted.

\begin{theorem}\label{Theorem_equivalence}
The following two statements (a) and (b) are
equivalent.
\begin{description}
\item[(a)] \ The function $v(t,x;w)=(x-w)^2e^{-\rho^2(T-t)}+\frac{\lambda \rho ^2}{4}\left(T^2-t^2\right)-\frac{\lambda}{2}\left(\rho^2T-\ln \frac{\sigma^2}{\pi \lambda}\right)(T-t)-(w-z)^2$, $(t,x)\in [0,T]\times \mathbb{R}$, is the optimal value function of the exploratory MV problem (\ref{value_function}), and the
%solves the HJB equation (\ref{HJB_second_case}) and the
corresponding optimal feedback control is
\begin{equation*}
\boldsymbol{\pi}^{\ast }(u;t,x,w)=\mathcal{N}\left( \, u\, \Big| -\frac{\rho}{\sigma}(x-w)\ , \  \frac{\lambda}{2\sigma^2}e^{\rho^2(T-t)} \right).
\end{equation*}%
%where $\{X_{t}^{\ast }, t\geq 0\}$ is the
%associated  optimal state process of the exploratory problem.

\item[(b)] \ The function $\omega(t,x;w)=(x-w)^2e^{-\rho^2(T-t)}-(w-z)^2$, $(t,x)\in [0,T]\times \mathbb{R}$, is the optimal value function of the classical MV problem (\ref{classical_value_function}), and the corresponding optimal feedback control is
\begin{equation*}
\boldsymbol{u}^{\ast }(t,x;w)=-\frac{\rho}{\sigma}(x-w).
\end{equation*}%
%where $\{x_{t}^{\ast },t\geq 0\}$
%is the associated optimal state process of the classical problem.
\end{description}
Moreover, the two problems have the same Lagrange multiplier $w=\frac{ze^{\rho^2T}-x_0}{e^{\rho^2T}-1}$.
\end{theorem}

It is reasonable to expect that the exploratory problem converges to its
classical counterpart as the exploration weight $\lambda$ decreases to 0. The following result makes this precise.

\begin{theorem}\label{convergence_to_Dirac}
Assume that statement (a) (or equivalently, (b)) of Theorem \ref{Theorem_equivalence} holds. Then, for each $(t,x,w)\in [0,T]\times \mathbb{R}\times \mathbb{R}$,
$$\lim_{\lambda \rightarrow 0}\boldsymbol{\pi}^{\ast}(\cdot;t,x;w)=\delta_{\boldsymbol{u}^{\ast}(t,x;w)}(\cdot) \;\;\mbox{ weakly.}$$
Moreover, %the associated value functions converge pointwise, i.e.,
$$\lim_{\lambda \rightarrow 0}|V(t,x;w)-V^{\text{cl}}(t,x;w)|=0.$$

\end{theorem}

\begin{proof}
The weak convergence of the feedback controls follows from the explicit forms of $\boldsymbol{\pi}^{\ast}$  and $\boldsymbol{u}^{\ast}$ in statements (a) and (b). The pointwise convergence of the value functions follows easily from the forms of $V(\cdot )$ and $%
V^{\text{cl}}(\cdot )$, together with the fact that
\[ \lim_{\lambda \rightarrow 0}\frac{\lambda}{2}\ln\frac{\sigma^2}{\pi \lambda} =0.
\]
\end{proof}

Finally, we conclude this section by examining the cost of the exploration. This was originally defined and derived in \cite{Hwang} for the  infinite horizon setting.  Here, the cost associated with the MV problem due to the explicit inclusion of exploration in the objective (\ref{value_function}) is defined by
$$\mathcal{C}^{u^{\ast },\pi ^{\ast }}(0,x_0;w):=\left( V(0,x_0;w)-{\lambda}\mathbb{E}\left[
\left. \int_{0}^{T }\int_{{\mathbb{R}}}\pi _{t}^{\ast }(u)\ln
\pi _{t}^{\ast }(u)du \,dt\right \vert X_{0}^{\pi^*}=x_0\right] \right)$$
\begin{equation}
-V^{\text{cl}}(0,x_0;w),
\label{exploration_cost}
\end{equation}
\smallskip
for $x_0\in \mathbb{R}$, where $\pi^*=\{\pi^*_t, t\in [0,T]\}$ is the (open-loop) optimal  strategy generated by the optimal feedback law $\boldsymbol{\pi}^{\ast }$ with respect to the initial condition $X_{0}^{\pi^*}=x_0$.
This cost is the difference between the two optimal value functions, adjusting for the additional contribution due to the entropy value of the optimal exploratory strategy.

Making use of Theorem \ref{Theorem_equivalence}, we have the following result.

\begin{theorem}\label{exploration_cost_theorem}
Assume that statement (a) (or equivalently, (b)) of Theorem \ref{Theorem_equivalence} holds. Then, the exploration cost for the MV problem is
\begin{equation}\label{LQ_cost_theorem}
\mathcal{C}^{u^{\ast },\pi ^{\ast }}(0,x_0;w)
%=V^{\text{cl}}(x)-\left( V(x)-\frac{\lambda}{2\rho }%
%\ln \left( \frac{2\pi e\lambda }{N-\alpha _{2}D^{2}}\right) \right)
=\frac{\lambda T}{2},\;\;x_0\in \mathbb{R},\;\;w\in \mathbb{R}.
\end{equation}
\end{theorem}
\begin{proof}
Let $\{\pi^{\ast }_t,t\in [0,T]\}$ be the open-loop control generated by the feedback control $\boldsymbol{\pi}^{\ast }$ given in statement (a) with respect to the initial state $x_0$ at $t=0$, namely,
\[
\pi^{\ast }_t(u)=\mathcal{N}\left( \, u\, \Big| -\frac{\rho}{\sigma}(X^*_t-w)\ , \  \frac{\lambda}{2\sigma^2}e^{\rho^2(T-t)} \right)
\]
where $\{X_{t}^{\ast }, t\in [0,T]\}$ is the
corresponding  optimal wealth process of the exploratory  MV problem, starting from the state $x_0$ at $t=0$, when $\boldsymbol{\pi}^{\ast }$ is applied.
Then, we easily deduce that
\[ \int_{\mathbb{R}}\pi^{\ast }_t(u)\ln \pi^{\ast }_t(u)du=-\frac{1}{2}\ln \left(\frac{\pi e \lambda}{\sigma^2}e^{\rho^2(T-t)}\right).\]
The desired result now follows immediately from the expressions of $V(\cdot )$ in (a) and $V^{\text{cl}}(\cdot)$ in (b).
\end{proof}
\medskip

The exploration cost depends only on two ``agent-specific" parameters, the exploration weight $\lambda>0$ and the investment horizon $T>0$. Note that the latter is also  the exploration horizon. Our result is intuitive in that the exploration cost increases with the exploration weight and  with the exploration horizon. Indeed,  the dependence is  {\it linear} with respect to each of the two attributes: $\lambda$ and $T$.\footnote{In
\cite{Hwang} with the infinite horizon LQ case, an analogous result is obtained which states that exploration cost is proportional to the exploration weight and inversely proportional to the discount factor. Clearly, here the length of time horizon, $T$, plays a role similar to the inverse of the discount factor.} It is also interesting to note that the cost is independent of the Lagrange multiplier. This suggests that the exploration cost will not increase when the agent is more aggressive (or risk-seeking) -- reflected by the expected target $z$ or equivalently the  Lagrange multiplier $w$.

\section{RL Algorithm Design}
Having laid the theoretical foundation in the previous two sections, we now  design an RL algorithm  to {\it learn} the solution of the entropy-regularized MV  problem and to output implementable portfolio allocation strategies, without assuming  any knowledge about the underlying parameters. To this end, we will first  establish a so-called {\it policy improvement theorem } as well as the corresponding convergence result. %that are provable for problem (\ref{value_function}).
Meanwhile, we will provide a self-correcting scheme to learn the true Lagrange multiplier $w$, based on stochastic approximation. Our RL algorithm bypasses the phase of estimating any model parameters, including the mean return vector and the variance-covariance matrix. It also avoids inverting a typically ill-conditioned variance-covariance matrix in high dimensions that would likely produce non-robust portfolio strategies.

In this paper, rather than relying on the typical framework of discrete-time MDP (that is used  for most RL problems) and discretizing time and space accordingly, we design an algorithm to learn the solutions of the continuous-time exploratory MV problem (\ref{value_function}) {\it directly}. Specifically, we adopt the approach  developed in \cite{Doya} to avoid  discretization of the state dynamics or the HJB equation. As pointed out in \cite{Doya}, it is typically challenging to find the right granularity to discretize the state, action and time, and
naive discretization  may lead to poor performance. On the other hand, grid-based discretization methods for solving the HJB equation cannot be easily extended to high-dimensional state space in practice due to the curse of dimensionality, although theoretical convergence results have been established (see \cite{Munos_1, Munos_2}). Our algorithm, to be described in Subsection 4.2, however, makes use of a provable policy improvement theorem and fairly simple yet efficient functional approximations to directly learn the value functions and the optimal Gaussian policy. Moreover, it is computationally feasible and implementable in high-dimensional state spaces (i.e., in the case of a large number of  risky assets) due to the explicit representation of the value functions and the portfolio strategies, thereby devoid of  the curse of dimensionality. Note that our algorithm does not use (deep) neural networks, which have been applied extensively in literature for (high-dimensional) continuous RL problems (e.g., \cite{Lillicrap}, \cite{DQN}) but known for unstable performance, sample inefficiency as well as extensive hyperparameter tuning (\cite{DQN}), in addition to their low interpretability.\footnote{Interpretability is one of the most important and pressing issues in the general artificial intelligence applications in financial industry due to, among others,
the regulatory requirement.}

\subsection{A policy improvement theorem}

Most RL algorithms consist of two iterative procedures: policy evaluation and policy improvement (\cite{SB}). The former  provides an estimated value function for the current policy, whereas the latter updates the current policy in the {\it right} direction to improve the value function.
A policy improvement theorem (PIT) is therefore
 a crucial prerequisite  for interpretable RL algorithms that ensures the iterated value functions to be non-increasing (in the case of a minimization problem), and ultimately converge to the optimal value function; see, for example, Section 4.2 of \cite{SB}.
PITs have been proved for discrete-time entropy-regularized RL problems in infinite horizon (\cite{Ha}), and for continuous-time classical stochastic control problems (\cite{PIT}).\footnote{\cite{PIT} studied classical stochastic control problems with no distributional controls nor entropy regularization. They did not consider  RL and related issues including exploration.} The following result provides a PIT for our  exploratory MV portfolio selection problem.

\begin{theorem}[Policy Improvement Theorem]\label{PIT}
 Let $w\in \mathbb{R}$ be fixed and  $\boldsymbol{\pi}=\boldsymbol{\pi}(\cdot;\cdot,\cdot,w)$ be an arbitrarily given admissible  feedback control policy. Suppose that the corresponding value function $V^{\boldsymbol{\pi}}(\cdot,\cdot;w)\in C^{1,2}([0,T)\times \mathbb{R})\cap C^0([0,T]\times \mathbb{R})$ and satisfies $V^{\boldsymbol{\pi}}_{xx}(t,x;w)>0$, for any $(t,x)\in [0,T)\times \mathbb{R}$. Suppose further that the feedback policy $\tilde{{\boldsymbol{\pi}}}$ defined by
\begin{equation}\label{new_policy}
\tilde{\boldsymbol{\pi}}(u;t,x,w)=\mathcal{N}\left( \, u\, \Big| -\frac{\rho}{\sigma}\frac{V^{\boldsymbol{\pi}}_x(t,x;w)}{V^{\boldsymbol{\pi}}_{xx}(t,x;w)}\ , \  \frac{\lambda}{\sigma^2V^{\boldsymbol{\pi}}_{xx}(t,x;w)} \right)
\end{equation}
is admissible. Then, %the corresponding value function $V^{\tilde{{\boldsymbol{\pi}}}}(t,x;w)$ satisfies
\begin{equation}\label{value_improve}
V^{\tilde{{\boldsymbol{\pi}}}}(t,x;w)\leq V^{\boldsymbol{\pi}}(t,x;w),\quad (t,x)\in [0,T]\times \mathbb{R}.
\end{equation}
\end{theorem}
\begin{proof}
Fix $(t,x)\in [0,T]\times \mathbb{R}$. Since, by assumption, the feedback policy $\tilde{{\boldsymbol{\pi}}}$ is admissible,
the open-loop control strategy, $\tilde{\pi}=\{\tilde{\pi}_v, v\in [t,T]\}$, generated from $\tilde{{\boldsymbol{\pi}}}$ with respect to the initial condition ${X}^{\tilde{{\boldsymbol{\pi}}}}_t=x$ is  admissible. Let $\{{X}^{\tilde{{\boldsymbol{\pi}}}}_s,s\in[t,T]\}$ be the corresponding wealth process under $\tilde{\pi}$. Applying It\^{o}'s formula, we have
$$V^{\boldsymbol{\pi}}(s,\tilde{X}_s)=V^{\boldsymbol{\pi}}(t,x)+\int_t^sV^{\boldsymbol{\pi}}_t(v,{X}^{\tilde{{\boldsymbol{\pi}}}}_v)dv+\int_t^s\int_{\mathbb{R}}\Big(\frac{1}{2}\sigma^2u^2V^{\boldsymbol{\pi}}_{xx}(v,{X}^{\tilde{{\boldsymbol{\pi}}}}_v)$$
\begin{equation}\label{Ito}
+\rho \sigma u V^{\boldsymbol{\pi}}_x(v,{X}^{\tilde{{\boldsymbol{\pi}}}}_v)\Big)\tilde{\pi}_v(u) \, dudv+\int_t^s \sigma \left(\int_{\mathbb{R}}u^2\tilde{\pi}_v(u)du\right)^{\frac{1}{2}}V^{\boldsymbol{\pi}}(v,{X}^{\tilde{{\boldsymbol{\pi}}}}_v)\, dW_v,\;s\in [t,T].
\end{equation}
Define the stopping times $\tau_n:=\inf\{s\geq t: \int_t^s \sigma^2\int_{\mathbb{R}}u^2\tilde{\pi}_v(u)du\left(V^{\boldsymbol{\pi}}(v,{X}^{\tilde{{\boldsymbol{\pi}}}}_v)\right)^2dv\geq n\}$, for $n\geq 1$. Then, from (\ref{Ito}), we obtain
$$V^{\boldsymbol{\pi}}(t,x)=\mathbb{E}\Big[V^{\boldsymbol{\pi}}(s\wedge\tau_n,{X}^{\tilde{{\boldsymbol{\pi}}}}_{s\wedge \tau_n})-\int_t^{s\wedge\tau_n}V^{\boldsymbol{\pi}}_t(v,{X}^{\tilde{{\boldsymbol{\pi}}}}_v)dv$$
\begin{equation}\label{stopped}
-\int_t^{s\wedge\tau_n}\int_{\mathbb{R}}\Big(\frac{1}{2}\sigma^2u^2V^{\boldsymbol{\pi}}_{xx}(v,{X}^{\tilde{{\boldsymbol{\pi}}}}_v)+\rho \sigma u V^{\boldsymbol{\pi}}_x(v,{X}^{\tilde{{\boldsymbol{\pi}}}}_v)\Big)\tilde{\pi}_v (u)\, dudv\, \Big| {X}^{\tilde{{\boldsymbol{\pi}}}}_t=x\Big].
\end{equation}
On the other hand, by standard arguments and the assumption that $V^{\boldsymbol{\pi}}$ is smooth, we have
$$V_t^{\boldsymbol{\pi}}(t,x)+\int_{\mathbb{R}}\left(\frac{1}{2}\sigma^2 u^2 V^{\boldsymbol{\pi}}_{xx}(t,x)+\rho\sigma u V^{\boldsymbol{\pi}}_x(t,x)+\lambda\ln\boldsymbol{\pi}(u;t,x)\right)\boldsymbol{\pi}(u;t,x)du=0,$$
for any $(t,x)\in [0,T)\times \mathbb{R}$. It follows that
\begin{equation}\label{minimization}
V_t^{\boldsymbol{\pi}}(t,x)+\min_{\pi'\in\mathcal{P}(\mathbb{R})}\int_{\mathbb{R}}\left(\frac{1}{2}\sigma^2 u^2 V^{\boldsymbol{\pi}}_{xx}(t,x)+\rho\sigma u V^{\boldsymbol{\pi}}_x(t,x)+\lambda\ln\pi'(u)\right)\pi'(u)du\leq 0.
\end{equation}
Notice that the minimizer of the Hamiltonian in (\ref{minimization}) is given by the feedback policy $\tilde{\boldsymbol{\pi}}$ in (\ref{new_policy}). It then follows that equation (\ref{stopped}) implies
$$V^{\boldsymbol{\pi}}(t,x)\geq\mathbb{E}\Big[V^{\boldsymbol{\pi}}(s\wedge\tau_n,{X}^{\tilde{{\boldsymbol{\pi}}}}_{s\wedge \tau_n})+\lambda \int_t^{s\wedge\tau_n}\int_{\mathbb{R}}\tilde{\pi}_v(u)\ln\tilde{\pi}_v(u)\, dudv\Big | {X}^{\tilde{{\boldsymbol{\pi}}}}_t=x\Big],$$
for $(t,x)\in [0,T]\times \mathbb{R}$ and $s\in [t,T]$. Now taking $s=T$, and using that $V^{\boldsymbol{\pi}}(T,x)=V^{\tilde{{\boldsymbol{\pi}}}}(T,x)=(x-w)^2-(w-z)^2$ together with the assumption that $\tilde{\pi}$ is admissible, we obtain, by sending  $n\rightarrow \infty$ and applying the dominated convergence theorem, that
$$V^{\boldsymbol{\pi}}(t,x)\geq\mathbb{E}\Big[V^{\tilde{{\boldsymbol{\pi}}}}(T,{X}^{\tilde{{\boldsymbol{\pi}}}}_{T})+\lambda \int_t^{T}\int_{\mathbb{R}}\tilde{\pi}_v(u)\ln\tilde{\pi}_v(u)\, dudv\Big | {X}^{\tilde{{\boldsymbol{\pi}}}}_t=x\Big]=V^{\tilde{{\boldsymbol{\pi}}}}(t,x),$$
for any $(t,x)\in [0,T]\times \mathbb{R}$.
\end{proof}

\medskip

The above theorem suggests that there are always policies in the Gaussian family
that improves
the value function of any given, not necessarily Gaussian, policy. Hence, without loss of generality, we can simply focus on
the Gaussian policies when choosing an initial solution. Moreover, the optimal Gaussian policy (\ref{Gaussian_verified}) in Theorem \ref{verification} suggests that a candidate initial feedback policy may take the form $\boldsymbol{\pi}_0(u;t,x,w)=\mathcal{N}(u| a(x-w), c_1e^{c_2(T-t)})$. It turns out that, theoretically,  such a choice leads to the convergence of both the value functions and the policies in a {\it finite} number of iterations. %, a consequence of exploiting the underlying linear--quadratic structure of the Mean-variance problem.

\begin{theorem}\label{convergence_learning}
Let $\boldsymbol{\pi}_0(u;t,x,w)=\mathcal{N}(u| a(x-w), c_1e^{c_2(T-t)})$, with $a,c_2\in \mathbb{R}$ and $c_1>0$. Denote by $\{\boldsymbol{\pi}_n(u;t,x,w), (t,x)\in [0,T]\times \mathbb{R},n\geq 1\}$ the sequence of feedback policies updated by the policy improvement scheme (\ref{new_policy}), and $\{V^{\boldsymbol{\pi}_n}(t,x;w), (t,x)\in [0,T]\times \mathbb{R}, n\geq 1\}$ the sequence of the corresponding value functions. Then,
\begin{equation}
\lim_{n\rightarrow \infty} \boldsymbol{\pi}_n (\cdot; t,x,w)= \boldsymbol{\pi^*}(\cdot; t,x,w) \;\;\mbox{ weakly,}
\end{equation}
and
\begin{equation}
\lim_{n\rightarrow \infty} V^{\boldsymbol{\pi}_n}(t,x;w)=V(t,x;w),
\end{equation}
for any $(t,x,w)\in [0,T]\times \mathbb{R}\times \mathbb{R}$, where $\boldsymbol{\pi^*}$ and $V$ are the optimal Gaussian policy (\ref{Gaussian_verified}) and the optimal value function (\ref{value_verified}), respectively.
\end{theorem}
\begin{proof}
It can be easily verified that the feedback policy $\boldsymbol{\pi}_0$ where $\boldsymbol{\pi}_0(u;t,x,w)=\mathcal{N}(u| a(x-w), c_1e^{c_2(T-t)})$ generates an open-loop policy $\pi_0$ that is admissible with respect to the initial $(t,x)$. Moreover, it follows from the Feynman-Kac formula that the corresponding value function $V^{\boldsymbol{\pi}_0}$ satisfies the PDE
$$V^{\boldsymbol{\pi}_0}_t(t,x;w)+\int_{\mathbb{R}}\Big(\frac{1}{2}\sigma^2 u^2 V^{\boldsymbol{\pi}_0}_{xx}(t,x;w)+\rho \sigma  u V^{\boldsymbol{\pi}_0}_x(t,x;w)$$
\begin{equation}
+\lambda \ln \boldsymbol{\pi}_0(u;t,x,w)\Big)\boldsymbol{\pi}_0(u;t,x,w)du=0,
\end{equation}
with terminal condition $V^{\boldsymbol{\pi}_0}(T,x;w)=(x-w)^2-(w-z)^2$.
Solving this equation we obtain
$$V^{\boldsymbol{\pi}_0}(t,x;w)=(x-w)^2e^{(2\rho\sigma a+\sigma^2 a^2)(T-t)}+\int_t^Tc_1\sigma^2e^{(2\rho\sigma a+\sigma^2 a^2+c_2)(T-s)}ds$$
\begin{equation}\label{first_value}
+\frac{\lambda c_2}{4}(T-t)^2+\frac{\lambda\ln(2\pi e c_1)}{2}(T-t)-(w-z)^2.
\end{equation}
It is easy to check  that $V^{\boldsymbol{\pi}_0}$ satisfies the conditions in Theorem \ref{PIT} and, hence, the theorem applies. The improved policy is given by (\ref{new_policy}), which, in the current case, becomes
$$\boldsymbol{\pi}_1(u;t,x,w)=\mathcal{N}\left( \, u\, \Big| -\frac{\rho}{\sigma}(x-w)\ , \  \frac{\lambda}{2\sigma^2e^{(2\rho \sigma a+\sigma^2 a^2)(T-t)}} \right).$$
Again, we can calculate the corresponding value function as $V^{\boldsymbol{\pi}_1}(t,x;w)=(x-w)^2e^{-\rho^2(T-t)}+F_1(t)$, where $F_1$ is a function of $t$ only. Theorem \ref{PIT} is applicable again, which yields the improved policy $\boldsymbol{\pi}_2$ as exactly the optimal Gaussian policy $\boldsymbol{\pi^*}$ given in (\ref{Gaussian_verified}), together with the optimal value function $V$ in (\ref{value_verified}). The desired convergence therefore follows, as for $n\geq 2$, both the policy and the value function will no longer strictly improve  under the policy improvement scheme (\ref{new_policy}).
\end{proof}

\medskip

The above convergence result shows that if we choose the initial policy wisely, the learning scheme will, {\it theoretically},  converge after a finite number of (two, in fact)  iterations. When implementing this scheme in practice, of course, the value function for each policy can only be approximated and, hence, the learning process typically takes more iterations to converge. Nevertheless, Theorem \ref{PIT} provides the theoretical foundation for  updating a  current policy, while Theorem \ref{convergence_learning} suggests a good starting point in the policy space. We will make use of both results in the next subsection to design an implementable RL algorithm for the exploratory MV problem.

\subsection{The EMV algorithm}

In this section, we present an RL algorithm, the EMV (exploratory mean--variance) algorithm, to solve (\ref{value_function}). It consists of three concurrently ongoing procedures: the policy evaluation, the policy improvement and a self-correcting scheme for learning the  Lagrange multiplier $w$ based on stochastic approximation.

For the policy evaluation, we follow the method employed in \cite{Doya} for learning the value function $V^{\boldsymbol{\pi}}$ under any  arbitrarily given  admissible feedback policy $\boldsymbol{\pi}$. By Bellman's consistency, we have
\begin{equation}\label{learning_value}
V^{\boldsymbol{\pi}}(t,x)=\mathbb{E}\left[V^{\boldsymbol{\pi}}(s, X_s)+\lambda\int_t^s\int_{\mathbb{R}}\pi_v(u)\ln \pi_v(u)dudv\big| X_t=x\right], \;s\in[t,T],
\end{equation}
for $(t,x)\in [0,T]\times \mathbb{R}$. Rearranging this equation and dividing both sides by $s-t$, we obtain
$$\mathbb{E}\left[\frac{V^{\boldsymbol{\pi}}(s,X_s)-V^{\boldsymbol{\pi}}(t,X_t)}{s-t}+\frac{\lambda}{s-t}\int_t^s\int_{\mathbb{R}}\pi_v(u)\ln \pi_v(u)dudv\big| X_t=x\right]=0.$$
Taking $s\rightarrow t$ gives rise to the continuous-time Bellman's error (or the {\it temproral difference} (TD) error; see \cite{Doya})
\begin{equation}
\delta_t:=\dot{V}^{\boldsymbol{\pi}}_t+\lambda\int_{\mathbb{R}} \pi_t(u)\ln \pi_t(u)du,
\end{equation}
where $\dot{V}^{\boldsymbol{\pi}}_t=\frac{V^{\boldsymbol{\pi}}(t+\Delta t,X_{t+\Delta t})-V^{\boldsymbol{\pi}}(t,X_t)}{\Delta t}$ is the total derivative and $\Delta t$ is the discretization step for the learning algorithm.%\footnote{{\bf To benefit people who are not familiar with RL, briefly explain here what this TD is for.}}

The objective of the policy evaluation procedure is to minimize the Bellman's error $\delta_t$. In general, this can be carried out as follows. Denote by $V^{\theta}$ and $\boldsymbol{\pi}^{\phi}$ respectively the parametrized value function and policy (upon using regressions or neural networks, or making use of some of
the structures of the problem; see below), with $\theta,\phi$ being the vector of weights to be learned. We then minimize
$$C(\theta,\phi)=\frac{1}{2}\mathbb{E}\left[\int_0^T |\delta_t|^2dt\right]=\frac{1}{2}\mathbb{E}\left[\int_0^T \big|\dot{V}^{\theta}_t+\lambda\int_{\mathbb{R}} \pi^{\phi}_t(u)\ln \pi^{\phi}_t(u)du\big|^2dt\right],$$
where ${\pi}^{\phi}=\{{\pi}^{\phi}_t, t\in [0,T]\}$ is  generated from $\boldsymbol{\pi}^{\phi}$ with respect to a given initial state $X_0=x_0$ at time 0.
To approximate $C(\theta,\phi)$ in an implementable algorithm, we first
discretize $[0,T]$ into small equal-length intervals $[t_i,t_{i+1}]$, $i=0,1,\cdots,l$, where $t_0=0$ and $t_{l+1}=T$.
Then we collect a set of samples $\mathcal{D}=\{(t_i,x_i), i=0,1,\cdots,l+1\}$ in the following way.
The initial sample is $(0,x_0)$ for $i=0$. Now, at each $t_i$, $i=0, 1,\cdots,l$, we
sample $\pi_{t_i}^{\phi}$ to obtain an allocation $u_i\in \mathbb{R}$ in the risky asset, and then observe the wealth $x_{i+1}$ at the next time instant $t_{i+1}$.
%according to the discretized version of the classical controlled wealth dynamics (\ref{classical_wealth}):
%$$x_{i+1}=x_i+\sigma u_i(\rho \Delta t+W_{t_{i+1}}-W_{t_i}).$$
We can now approximate $C(\theta,\phi)$ by
\begin{equation}\label{learning_objective}
C(\theta,\phi)=\frac{1}{2}\sum_{(t_i,x_i)\in \mathcal{D}}\left(\dot{V}^{\theta}(t_i,x_i)+\lambda\int_{\mathbb{R}} \pi^{\phi}_{t_i}(u)\ln \pi^{\phi}_{t_i}(u)du\right)^2 \Delta t.
\end{equation}
%Here, the set $\mathcal{D}$ contains the collected samples $(t_i,x_i)$ through sampling the allocation $u_i\in \mathbb{R}$ in the risky asset under $\boldsymbol{\pi}^{\phi}$ at $\pi_{t_i}^{\phi}$ at each $t_i$. In simulations, the observed wealth changes over each interval $[t_i,t_{i+1}]$ according to the discretized version of the classical controlled wealth dynamics (\ref{classical_wealth})
%$$x_{i+1}=x_i+\sigma u_i(\rho \Delta t+W_{t_{i+1}}-W_{t_i}).$$

Instead of following  the common practice to represent $V^{\theta}$, $\pi^{\phi}$ using (deep) neural networks for continuous RL problems, in this paper we will take advantage of the more explicit parametric expressions obtained  in Theorem \ref{verification} and Theorem \ref{convergence_learning}. This will lead to faster learning and convergence, which will be demonstrated in all the numerical experiments below. More precisely, by virtue of  Theorem \ref{convergence_learning}, we will focus on Gaussian policies with variance taking the form $c_1e^{c_2(T-t)}$, which in turn leads to the entropy parametrized by $\mathcal{H}(\pi^{\phi}_t)=\phi_1+\phi_2(T-t)$, where $\phi=(\phi_1,\phi_2)'$, with $\phi_1\in \mathbb{R}$ and $\phi_2>0$, is the parameter vector to be learned.

On the other hand, as suggested by the theoretical  optimal value function (\ref{value_verified}) in Theorem \ref{verification}, we consider the parameterized  $V^{\theta}$, where $\theta=(\theta_0,\theta_1,\theta_2,\theta_3)'$, by
\begin{equation}\label{V_parametrized}
 V^{\theta}(t,x)=(x-w)^2e^{-\theta_3(T-t)}+\theta_2 t^2+\theta_1 t+\theta_0,\;\;(t,x)\in [0,T]\times \mathbb{R}.
 \end{equation}
 From the policy improvement updating scheme (\ref{new_policy}), it follows that the variance of the policy $\pi^{\phi}_t$ is $\frac{\lambda}{2\sigma^2}e^{\theta_3(T-t)}$, resulting in  the entropy $\frac{1}{2}\ln \frac{\pi e \lambda}{\sigma^2}+\frac{\theta_3}{2}(T-t)$.
 Equating this with the previously derived form  $\mathcal{H}(\pi^{\phi}_t)=\phi_1+\phi_2(T-t)$, we deduce
\begin{equation}
\sigma^2=\lambda\pi e^{1-2\phi_1}\quad \text{and}\quad \theta_3=2\phi_2=\rho^2.
\end{equation}
The improved policy, in turn, becomes, according to (\ref{new_policy}), that
$$\boldsymbol{\pi}(u;t,x,w)=\mathcal{N}\left( \, u\, \Big| -\frac{\rho}{\sigma}(x-w)\ ,\ \frac{\lambda}{2\sigma^2}e^{\theta_3(T-t)}\right)$$
\begin{equation}\label{policy_updating}
=\mathcal{N}\left( \, u\, \Big|-\sqrt{\frac{2\phi_2}{\lambda\pi}}e^{\frac{2\phi_1-1}{2}}(x-w)\ ,\ \frac{1}{2\pi}e^{2\phi_2(T-t)+2\phi_1-1}\right),
\end{equation}
where we have assumed that the true (unknown) Sharpe ratio $\rho>0$.

Rewriting the objective (\ref{learning_objective}) using $\mathcal{H}(\pi^{\phi}_t)=\phi_1+\phi_2(T-t)$, we obtain
$$C(\theta,\phi)=\frac{1}{2}\sum_{(t_i,x_i)\in \mathcal{D}}\left(\dot{V}^{\theta}(t_i,x_i)-\lambda (\phi_1+\phi_2(T-t))\right)^2\Delta t.$$
Note that $\dot{V}^{\theta}(t_i,x_i)=\frac{{V}^{\theta}(t_{i+1},x_{i+1})-{V}^{\theta}(t_i,x_i)}{\Delta t}$, with $\theta_3=2\phi_2$ in the parametrization of ${V}^{\theta}(t_i,x_i)$. It is now straightforward to devise the updating rules for $(\theta_1,\theta_2)'$ and $(\phi_1,\phi_2)'$ using stochastic gradient descent algorithms (see, for example, Chapter $8$ of \cite{SGD}). Precisely, we compute
\begin{equation}\label{theta_1}
\frac{\partial C}{\partial \theta_1}=\sum_{(t_i,x_i)\in \mathcal{D}}\left(\dot{V}^{\theta}(t_i,x_i)-\lambda (\phi_1+\phi_2(T-t_i))\right)\Delta t;
\end{equation}
\begin{equation}\label{theta_2}
\frac{\partial C}{\partial \theta_2}=\sum_{(t_i,x_i)\in \mathcal{D}}\left(\dot{V}^{\theta}(t_i,x_i)-\lambda (\phi_1+\phi_2(T-t_i))\right)(t_{i+1}^2-t^2_i);
\end{equation}
\begin{equation}\label{phi_1}
\frac{\partial C}{\partial \phi_1}=-\lambda\sum_{(t_i,x_i)\in \mathcal{D}}\left(\dot{V}^{\theta}(t_i,x_i)-\lambda (\phi_1+\phi_2(T-t_i))\right)\Delta t;
\end{equation}
$$\frac{\partial C}{\partial \phi_2}=\sum_{(t_i,x_i)\in \mathcal{D}}\left(\dot{V}^{\theta}(t_i,x_i)-\lambda (\phi_1+\phi_2(T-t_i))\right)\Delta t$$
\begin{equation}\label{phi_2}
\times\left(-\frac{2(x_{i+1}-w)^2e^{-2\phi_2(T-t_{i+1})}(T-t_{i+1})-2(x_{i}-w)^2e^{-2\phi_2(T-t_{i})}(T-t_{i})}{\Delta t}-\lambda (T-t_i)\right).
\end{equation}
Moreover, the parameter $\theta_3$ is updated with $\theta_3=2\phi_2$, and
$\theta_0$ is updated based on the terminal condition $V^{\theta}(T,x;w)=(x-w)^2-(w-z)^2$, which yields
\begin{equation}\label{theta_0}
\theta_0=-\theta_2T^2-\theta_1T-(w-z)^2.
\end{equation}

Finally, we provide a scheme for learning the underlying Lagrange multiplier $w$. Indeed, the constraint $\mathbb{E}[X_T]=z$ itself suggests the standard stochastic approximation update
\begin{equation}\label{w_update}
w_{n+1}=w_n-\alpha_n(X_T-z),
\end{equation}
with $\alpha_n>0$, $n\geq 1$, being the learning rate. In implementation, we can replace $X_T$ in (\ref{w_update}) by a sample average $\frac{1}{N}\sum_j x^j_T$ to have a more stable learning process (see, for example, Section 1.1 of \cite{kushner}), where $N\geq 1$ is the sample size and  $x^j_T$'s are the most recent $N$ terminal wealth values obtained at the time when $w$ is to be updated. It is interesting to notice that the learning scheme of $w$, (\ref{w_update}), is statistically  self-correcting. For example, if the (sample average) terminal wealth is above the target $z$, the updating rule (\ref{w_update}) will decrease $w$, which in turn decreases the mean of the exploratory Gaussian policy $\boldsymbol{\pi}$ in view of (\ref{policy_updating}). This implies that there will be less risky allocation in the next step of actions for learning and optimizing, leading to, on average, a decreased terminal wealth.\footnote{This discussion here is based on the assumption that the market has a positive Sharpe ratio (recall (\ref{policy_updating})). The case of a negative Sharpe ratio can be dealt with similarly.}

\medskip

We now summarize the pseudocode for the EMV algorithm.

\begin{algorithm}
\floatname{algorithm}{Algorithm $1$}
\renewcommand{\thealgorithm}{}
\caption{\textit{EMV: Exploratory Mean-Variance Portfolio Selection}}
\label{protocol1}
\begin{algorithmic}
\STATE  \textbf{Input}: Market Simulator \textit{Market}, learning rates $\alpha,\eta_{\theta},\eta_{\phi}$, initial wealth $x_0$, target payoff $z$, investment horizon $T$, discretization $\Delta t$, exploration rate  $\lambda$, number of iterations $M$, sample average size $N$.
\STATE  Initialize $\theta$, $\phi$ and $w$
\STATE  \textbf{for} $k=1$ \textbf{to} $M$ \textbf{do}
\STATE  \quad \textbf{for} $i=1$ \textbf{to} $\floor{\frac{T}{\Delta t}}$ \textbf{do}
\STATE \quad \quad Sample $(t^k_i,x^k_i)$ from \textit{Market} under $\pi^{\phi}$
\STATE \quad \quad Obtain collected samples $\mathcal{D}=\{(t^k_i,x^k_i), 1\leq i\leq  \floor{\frac{T}{\Delta t}}\}$
\STATE \quad \quad Update $\theta\leftarrow\theta-\eta_{\theta}\nabla_{\theta} C(\theta,\phi)$ using (\ref{theta_1}) and (\ref{theta_2})
\STATE \quad \quad Update $\theta_0$ using (\ref{theta_0}) and $\theta_3\leftarrow 2\phi_2$
\STATE \quad \quad Update $\phi\leftarrow\phi-\eta_{\phi}\nabla_{\phi} C(\theta,\phi)$ using (\ref{phi_1}) and (\ref{phi_2})
\STATE \quad \textbf{end for}
\STATE \quad Update $\pi^{\phi}\leftarrow\mathcal{N}\left( \, u\, \Big|-\sqrt{\frac{2\phi_2}{\lambda\pi}}e^{\frac{2\phi_1-1}{2}}(x-w)\ ,\ \frac{1}{2\pi}e^{2\phi_2(T-t)+2\phi_1-1}\right)$
\STATE \quad \textbf{if} $k$ mod $N==0$
\STATE \quad \quad Update $w\leftarrow w-\alpha\left(\frac{1}{N}\sum_{j=k-N+1}^k x^j_{\floor{\frac{T}{\Delta t}}}-z\right)$
\STATE \quad \textbf{end if}
\STATE \textbf{end for}
\end{algorithmic}
\end{algorithm}

\section{Simulation Results}
In this section, we compare the performance of our RL algorithm EMV with two other methods that could be used to solve the classical MV problem (\ref{target}). The first one is the traditional maximum likelihood estimation (MLE) that relies on the real-time estimation of the drift $\mu$ and the volatility $\sigma$ in the geometric Brownian motion price model (\ref{price}). Once the estimators of $\mu$ and $\sigma$ are available using the most recent price time series, the portfolio  allocation can be computed using the optimal allocation (\ref{classical_optimal_control}) for the classical MV problem. Another alternative  is based on the deep deterministic policy gradient (DDPG) method developed by \cite{Lillicrap}, a method that has been taken as the baseline approach for comparisons of different RL algorithms that solve continuous control problems.

We conduct comparisons in various simulation settings, including the stationary setting where the price of the risky asset evolves according to the geometric Brownian motion (\ref{price}) and the non-stationary setting that involves stochastic factor modeling for the drift and volatility parameters. We also consider a time-decreasing exploration case which can be easily accommodated by the EMV through an annealing $\lambda$ across all the episodes $[0,T]$ in the learning process. In nearly all of the experiments, our EMV algorithm outperforms the other two methods by large margins.

\medskip

In the following, we briefly describe the two other methods.

\smallskip

\noindent \textbf{\textit{Maximum likelihood estimation (MLE)}}\\
\noindent The MLE is a popular method for estimating the parameters $\mu$ and $\sigma$ in the geometric Brownian motion  model (\ref{price}). We refer the interested readers to Section 9.3.2 of \cite{MLE} for a detailed description of this method. At each decision making time $t_i$, the MLE estimators for $\mu$ and $\sigma$ are calculated based on the most recent $100$ data points of the price. One can then substitute the MLE estimators into the optimal allocation (\ref{classical_optimal_control}) and the expression of the Lagrange multiplier $w=\frac{ze^{\rho^2T}-x_0}{e^{\rho^2T}-1}$ to compute the allocation $u_i\in \mathbb{R}$ in the risky asset. Such a two-phase procedure is commonly used in adaptive control, where the first phase is  identification   and the second one is optimization (see, for example, \cite{CG}, \cite{Kumar}). The real-time estimation procedure also allows the MLE to be applicable in non-stationary markets with time-varying $\mu$ and $\sigma$.

\medskip

\noindent \textbf{\textit{Deep deterministic policy gradient (DDPG)}}\\
\noindent The DDPG method has attracted significant attention since it was introduced  in \cite{Lillicrap}. It has been taken as a
state-of-the-art baseline approach for continuous control (action) RL problems, albeit  in discrete time. The DDPG learns a deterministic target policy using deep neural networks for both the critic and the actor, with exogenous noise being added to encourage exploration (e.g. using OU processes; see \cite{Lillicrap} for details). To adapt the DDPG to the classical MV setting (without entropy regularization), we make the following adjustments. Since the target policy we aim to learn is a deterministic function of $(x-w)$
(see (\ref{classical_optimal_control})), we will  feed the samples $x_{i}-w$, rather than only $x_{i}$, to the actor network in order to output the allocation $u_i\in \mathbb{R}$. Here, $w$ is the learned Lagrange multiplier at the decision making time $t_i$, obtained from the same self-correcting scheme (\ref{w_update}). This modification also enables us to connect current allocation with the previously obtained sample average of terminal wealth in a feedback loop, through the Lagrange multiplier $w$. Another modification  from the original DDPG is that we include prioritized experience replay (\cite{prioritized}), rather than sampling experience uniformly from the replay buffer. We select the terminal experience with higher probability to train the critic and actor networks, to account for the fact that the  MV problem has no running cost, but only a terminal cost given by $(x_T-w)^2-(w-z)$ (cf. (\ref{unconstrained_classical})). Such a modification significantly improves learning speed and performance.

\subsection{The stationary market case}

We first perform numerical simulations in a stationary market environment, where the price process is simulated according to the geometric Brownian motion (\ref{price}) with constant $\mu$ and $\sigma$. We take $T=1$ and $\Delta t=\frac{1}{252}$, indicating that the MV problem is considered over a one-year period, with daily rebalancing. Reasonable values of the annualized return and volatility will be taken from the sets $\mu\in \{-50\%, -30\%,-10\%, 0\%,10\%, 30\%, 50\%\}$ and $\sigma \in \{10\%, 20\%,30\%, 40\%\}$, respectively. These values are usually considered for a ``typical" stock for simulation purpose (see, for example, \cite{Lo}). The annualized interest rate is taken to be $r=2\%$. We consider the  MV problem with a $40\%$ annualized target return on the terminal wealth starting from a normalized initial wealth $x_0=1$ and, hence, $z=1.4$ in (\ref{target}). These model parameters will be fixed for all the simulations considered in this paper.

For the EMV algorithm, we take the total training episodes $M=20000$ and the sample  size $N=10$ for learning the Lagrange multiplier $w$. The temperature parameter $\lambda=2$. Across all the simulations in this section, the learning rates are fixed as $\alpha=0.05$ and $\eta_{\theta}=\eta_{\phi}=0.0005$.

We choose the same $M$ and $N$ for the DDPG algorithm for a fair comparison. The critic network has $3$ hidden layers with $10, 8, 8$ hidden units for each layer, and the actor network has $2$ hidden layers with $10, 8$ hidden units. The learning rates for the critic and actor are $0.0001$, fixed across all simulations in this section as the case for EMV. The replay buffer has size $80$, and the minibatch size is $20$ for the stochastic gradient decent. The target network has the soft update parameter $\tau=0.001$. Finally, we adopt the OU process for adding exploration noise; see \cite{Lillicrap} for more details. All the simulations using the DDPG algorithm were trained in Tensorflow.

We summarize in Table \ref{table_1} the simulation results for the three methods, EMV, MLE and DDPG, under market scenarios corresponding to different combinations of $\mu$'s and $\sigma$'s.   For each method under each market scenario, we present the annualized sample mean $\overline{M}$ and sample variance $\overline{V}$ of the last $2000$ terminal wealth,  and the corresponding annualized Sharpe ratio $SR=\frac{\overline{M}-1}{\sqrt{\overline{V}}}$.

\afterpage{
\setlength\LTleft{-0.6in}
\setlength\LTright{-1in}
\begin{longtable}{|l|l|l|l|}
%\begin{adjustwidth}{–1cm}{–1cm}
\caption[Caption for LOF]{Comparison of the annualized sample mean ($\overline{M}$), sample variance ($\overline{V}$), Sharpe ratio ($SR$) and average training time (per experiment) for EMV, MLE and DDPG.\footnotemark}\label{table_1} \\
\hline
\textbf{Market scenarios} & \textbf{\ \ \ \ \ \ \ \ \  EMV} & \textbf{\ \ \ \ \ \ \ \ \ MLE} & \textbf{\ \ \ \ \ \ \ \ DDPG} \\
\hline
\endfirsthead
\multicolumn{4}{l}%
{\tablename\ \thetable\ -- \textit{Continued from previous page}} \\
\hline
\textbf{Market scenarios} & \textbf{\ \ \ \ \ \  EMV} & \textbf{\ \ \ \ \ \ \ \ \ MLE} & \textbf{\ \ \ \ \ \ \ \ DDPG} \\
\hline
\endhead
\hline \multicolumn{4}{r}{\textit{Continued on next page}} \\
\endfoot
\hline
\endlastfoot
$\mu=-50\%, \sigma=10\%$ & 1.396; \ \ 0.006; \ \ 5.107 & 1.556; \ \ 0.017; \ \ 4.284 & 1.297; \ \ 0.107; \ \ 0.908 \\
$\mu=-30\%,	\sigma=10\%$ & 1.390; \ \ 0.016; \ \ 3.039 & 1.215; \ \ 0.014;  \ \ 1.833 & 1.401; \ \ 0.003; \ \ 7.076 \\
$\mu=-10\%,	\sigma=10\%$ & 1.330; \ \ 0.074; \ \ 1.218 & 1.056; \ \ 1.365; \ \ 0.482 & 0.901; \ \ 0.014; \ \ -0.833 \\
$\mu=0\%,	\quad \ \sigma=10\%$ &  1.204; \ \ 1.280; \ \ 0.180 & 0.926; \ \ 38.40; \ \, -0.012  & 1.029; \ \ 0.038; \ \ 0.147 \\
$\mu=10\%,	\ \ \, \sigma=10\%$ & 1.318; \ \ 0.171; \ \ 0.769 & 1.009; \ \ 0.453; \ \  0.014 & 0.951; \ \  0.008; \ \ -0.541 \\
$\mu=30\%,	\ \ \, \sigma=10\%$ & 1.385; \ \ 0.019; \ \ 2.785 & 1.179; \ \ 0.059; \ \  0.737 & 0.224; \ \  0.104; \ \ -2.405 \\
$\mu=50\%,	\ \ \, \sigma=10\%$ & 1.394; \ \ 0.007; \ \ 4.772 & 1.459; \ \ 0.013; \ \ 3.983 & 1.478; \ \ 0.667; \ \ 0.717 \\
$\mu=-50\%,	\sigma=20\%$ & 1.387; \ \ 0.022; \ \ 2.606 & 1.310; \ \ 0.050; \ \ 1.387 & -0.551; \ 0.425; \ \ -2.379 \\
$\mu=-30\%,	\sigma=20\%$ & 1.359; \ \ 0.051; \ \  1.598 & 1.205; \ \ 1.319; \ \ 0.178 & 1.973; \ \ 0.404; \ \ 1.531 \\
$\mu=-10\%,	\sigma=20\%$ & 1.309; \ \ 0.245; \ \  0.625 & 1.036; \ \ 2.183; \ \ 0.024 & 1.349; \ \ 0.368; \ \ 0.575 \\
$\mu=0\%,	\quad \ \sigma=20\%$ & 1.105; \ \ 0.727; \ \ 0.123 & 0.921; \ \ 6.887; \ \ -0.301  & 0.988; \ \ 0.139; \ \ -0.033 \\
$\mu=10\%,	\ \ \, \sigma=20\%$ & 1.221; \ \ 0.314; \ \ 0.395 & 1.045; \ \ 6.751; \ \ 0.017 & 1.243; \ \ 0.354; \ \ 0.408 \\
$\mu=30\%,	\ \ \, \sigma=20\%$ & 1.345; \ \ 0.062; \ \ 1.387 & 1.155; \ \ 1.743; \ \ 0.117 & 1.360; \ \ 0.050; \ \ 1.613 \\
$\mu=50\%,	\ \ \, \sigma=20\%$ & 1.385; \ \ 0.027; \ \ 2.350 & 1.237; \ \ 1.293; \ \ 0.208 & 1.385; \ \ 0.004; \ \ 6.496 \\
$\mu=-50\%,	\sigma=30\%$ & 1.353; \ \ 0.044; \ \  1.682 & 1.333; \ \ 9.465; \ \ 0.108  & 0.272; \ \ 2.762; \ \ -0.438 \\
$\mu=-30\%,	\sigma=30\%$ & 1.323; \ \ 0.106; \ \  0.992 & 1.092; \ \ 5.657; \ \  0.039 & 0.034; \ \ 0.924; \ \ -1.005 \\
$\mu=-10\%,	\sigma=30\%$ & 1.317; \ \ 0.696; \ \  0.380 & 1.045; \ \ 17.87; \ \  0.011 & 1.371; \ \ 0.792; \ \ 0.417 \\
$\mu=0\%,	\quad \ \sigma=30\%$ & 1.079; \ \ 0.727; \ \ 0.092 & 0.955; \ \ 28.84; \ \ -0.008 & 1.070; \ \ 0.752; \ \ 0.081 \\
$\mu=10\%,	\ \ \, \sigma=30\%$ & 1.282; \ \ 0.885; \ \ 0.300 & 0.885; \ \ 24.06; \ \ -0.023 & 1.243; \ \ 0.825; \ \ 0.268 \\
$\mu=30\%,	\ \ \, \sigma=30\%$ & 1.334; \ \ 0.131; \ \ 0.921 & 0.886; \ \ 24.41; \ \ -0.023 & 1.210; \ \ 0.921; \ \ 0.218 \\
$\mu=50\%,	\ \ \, \sigma=30\%$ & 1.350; \ \ 0.049; \ \ 1.583 & 1.238; \ \  7.505; \ \ 0.087 & 0.610; \ \ 0.143 \ \ \ -1.030  \\
$\mu=-50\%,	\sigma=40\%$ & 1.342; \ \ 0.061; \ \ 1.385 & 1.284; \ \  11.14; \ \  0.085 & 1.328; \ \ 0.501; \ \ 0.463 \\
$\mu=-30\%,	\sigma=40\%$ & 1.320; \ \ 0.146; \ \ 0.839 & 1.145; \ \ 3.315; \ \  0.080 & 1.212; \ \ 0.160; \ \  0.531\\
$\mu=-10\%,	\sigma=40\%$ & 1.241; \ \ 0.707; \ \ 0.287 & 0.979; \ \ 7.960; \ \  -0.007 & 1.335; \ \ 1.413; \ \  0.282 \\
$\mu=0\%,	\quad \ \sigma=40\%$ & 1.057; \ \ 0.671; \ \  0.070 & 0.950; \ \ 31.60; \ \ -0.009 & 1.064; \ \ 1.467; \ \ 0.053  \\
$\mu=10\%,	\ \ \, \sigma=40\%$ & 1.155; \ \ 0.591; \ \ 0.202 & 1.053; \ \ 9.090;  \ \  0.017 & 1.242; \ \ 1.499; \ \ 0.198 \\
$\mu=30\%,	\ \ \, \sigma=40\%$ & 1.320; \ \ 0.198; \ \ 0.716 & 1.083; \ \ 17.46;  \ \  0.020 & 0.179; \ \ 1.533; \ \ -0.663 \\
$\mu=50\%,	\ \ \, \sigma=40\%$ & 1.329; \ \ 0.078; \ \ 1.174 & 0.963; \ \ 43.17; \ \ -0.006  & -0.390; \ 1.577;  \ \ -1.107 \\
\hline
\textbf{\ \ \ Training time} & $\qquad \quad \ <10$s & $\qquad \quad \ \ <10$s & $\qquad \quad \ \ \approx 3$hrs \\
\end{longtable}
%\end{adjustwidth}
%\restoregeometry
\footnotetext{If the hyperparameters are allowed to be tuned on a scenario-by-scenario basis (i.e., not for comparison purpose), the performance of EMV can be further improved with the Shape ratio close to its theoretical maximum, whereas improvement is more difficult to achieve for DDPG by hyperparameter tuning.}
\clearpage
}
A few observations are in order. First of all, the EMV algorithm outperforms the other two methods by a large margin in  several statistics of the investment outcome including sample mean, sample variance and Sharpe ratio. In fact, based on the comparison of Sharpe ratios, EMV outperforms  MLE in all the $28$ experiments, and outperforms DDPG  in $23$ out of the total $28$ experiments. Notice that  DDPG  yields rather unstable performance across different market scenarios, with some of the sample average terminal wealth below $0$ indicating the occurrence of bankruptcy. The EMV algorithm, on the other hand, achieves positive annualized return in all the experiments. The advantage  of the EMV algorithm over the deep learning DDPG algorithm is even more
 significant if we take into account the training time (all the experiments were performed on a MacBook Air laptop). Indeed,  DDPG involves extensive training of two deep neural networks, making it less appealing for high-frequency portfolio rebalancing and trading in practice.

Another advantage of EMV over DDPG is the ease of hyperparameter tuning. Recall that the learning rates are fixed respectively for EMV and DDPG across all the experiments for different $\mu$'s and $\sigma$'s. The performance of the EMV algorithm is less affected by the fixed learning rates, while the DDPG algorithm is not. Indeed, DDPG has been noted for its notorious brittleness and hyperparameter sensitivity in practice (see, for example, \cite{Deep_RL_1}, \cite{Deep_RL_2}), issues also shared by other deep RL methods for continuous control problems. Our EMV algorithm does not suffer from such issues, as a result of avoiding deep neural networks in its training and decision making processes.

The MLE method, although free from any learning rates tuning, needs to estimate the underlying parameters $\mu$ and $\sigma$. In all the simulations presented in Table \ref{table_1}, the estimated value of $\sigma$ is relatively close to its true value, while the drift parameter $\mu$ cannot be estimated accurately. This is consistent with the well documented mean--blur problem, which in turn leads to higher variance of the terminal wealth (see Table \ref{table_1}) when one applies the estimated $\mu$ and $\sigma$ to select the risky allocation (\ref{classical_optimal_control}).

Finally, we present the learning curves for the three methods in Figure \ref{fig_mean} and Figure \ref{fig_variance}. We plot the changes of the sample mean and sample variance of every (non-overlapping) $50$ terminal wealth, as the learning proceeds for each method.\footnote{Note the log scale for the variance in Figure \ref{fig_variance}.} From these plots, we can see that the EMV algorithm converges relatively faster than the other two methods, achieving relatively good performance even in the early phase of the learning process. This is also consistent with the convergence result in Theorem \ref{convergence_learning} and the remarks immediately following it in Section 4.1.

\begin{figure}[H]
\begin{center}
   \includegraphics[scale=0.25]{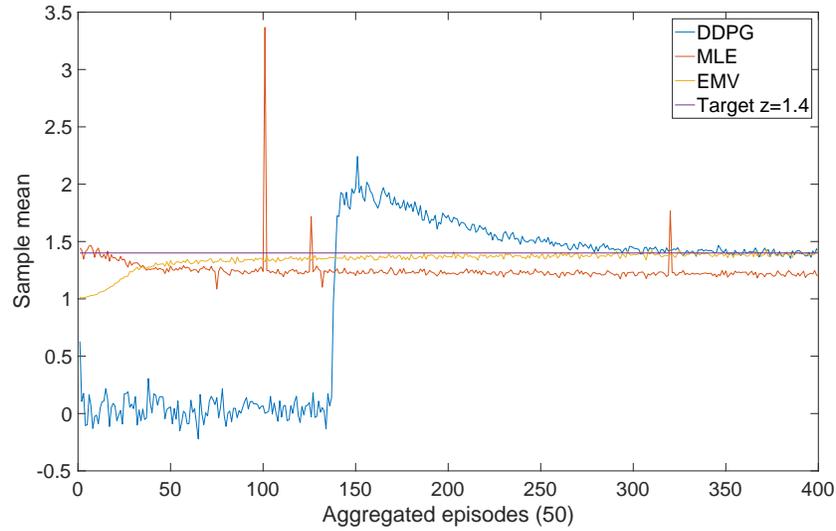}
   \caption{Learning curves of sample means of terminal wealth (over every $50$ iterations) for EMV, MLE and DDPG ($\mu=-30\%,	\sigma=10\%$).}\label{fig_mean}
\end{center}
\end{figure}

\begin{figure}[H]
\begin{center}
   \includegraphics[scale=0.25]{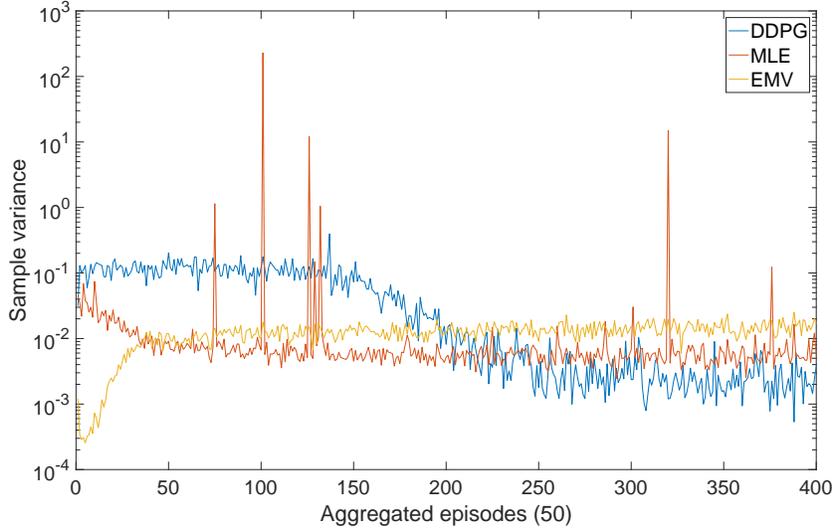}
   \caption{Learning curves of sample variances of terminal wealth (over every $50$ iterations) for EMV, MLE and DDPG ($\mu=-30\%,	\sigma=10\%$).}\label{fig_variance}
\end{center}
\end{figure}
\subsection{The non-stationary market case}
When it comes to the application domains of RL, 
one of the major differences between quantitative finance and other domains (such as AlphaGo; see \cite{Go1}) is that the underlying unknown investment environment is typically time-varying with the former. In this subsection, we consider the performance of the previous three methods in the non-stationary scenario where the price process is modeled by a stochastic factor model. To have a well-defined learning problem, the stochastic factor needs to change at a much slower time-scale, compared to that of the learning process. Specifically, we take the slow factor model within a multi-scale stochastic volatility model framework (see, for example, \cite{fouque}). The price process follows
\begin{equation*}
dS_t=S_t(\mu_t\, dt+\sigma \,dW_t),\quad 0<t\leq T, \quad S_0=s_0>0,
\end{equation*}
with $\mu_t,\sigma_t$, $t\in [0,T]$, being respectively the drift and volatility processes restricted to each simulation episode $[0,T]$ (so they may vary across different episodes). The controlled wealth dynamics over each episode $[0,T]$ is therefore
\begin{equation}\label{wealth_time_varying}
dx_t^u=\sigma_t u_t(\rho_t\, dt+\, dW_t), \quad 0<t\leq T, \quad x^u_0=x_0\in \mathbb{R}.
\end{equation}
For simulations in this subsection, we take
\begin{equation}\label{stochastic_factor}
d\rho_t=\delta dt \quad \text{and}\quad d\sigma_t=\sigma_t(\delta \, dt+\sqrt{\delta}\, dW^1_t), \quad 0<t\leq MT,
\end{equation}
with $\rho_0\in \mathbb{R}$ and $\sigma_0>0$ being the initial values respectively, $\delta>0$ being a small parameter, and the two Brownian motions satisfying $d\langle W, W^1\rangle_t=\gamma dt$, with $|\gamma|<1$. Notice that the terminal horizon for the stochastic factor model (\ref{stochastic_factor}) is $MT$ (recall $M=20000$ and $T=1$ as in Subsection 5.1), indicating  that the stochastic factors $\rho_t$ and $\sigma_t$ change across all the $M$ episodes. To make sure that their values stay in reasonable ranges after running for all $M$ episodes, we take $\delta$ to be small with value $\delta=0.0001$, and $\rho_0=-3.2$, $\sigma_0=10\%$, corresponding to the case $\mu_0=-30\%$ initially. We plot the learning curves in Figure \ref{fig_mean_non} and Figure \ref{fig_variance_non}. Clearly, the EMV algorithm displays remarkable stability for learning performance, compared to the other two methods even in the non-stationary market scenario.

\begin{figure}[H]
\begin{center}
   \includegraphics[scale=0.25]{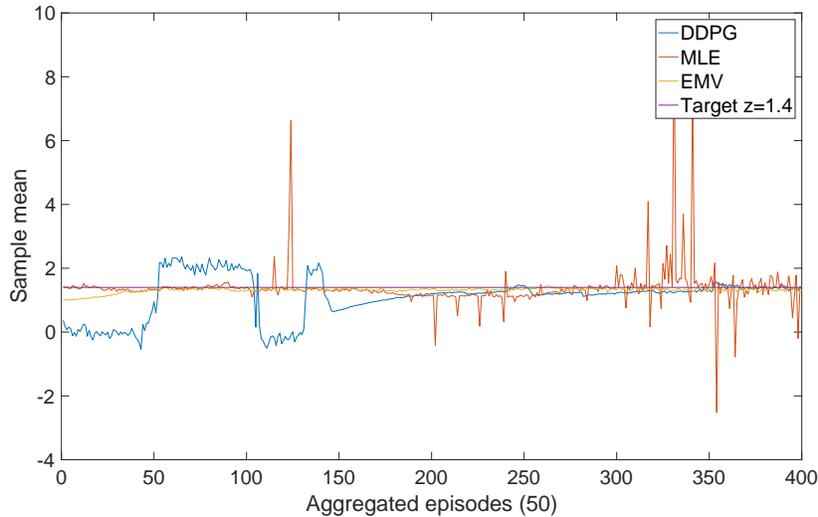}
   \caption{Learning curves of sample means of terminal wealth (over every $50$ iterations) for EMV, MLE and DDPG for non-stationary market scenario ($\mu_0=-30\%,	\sigma_0=10\%, \delta=0.0001, \gamma =0$).}\label{fig_mean_non}
\end{center}
\end{figure}

\begin{figure}[H]
\begin{center}
   \includegraphics[scale=0.25]{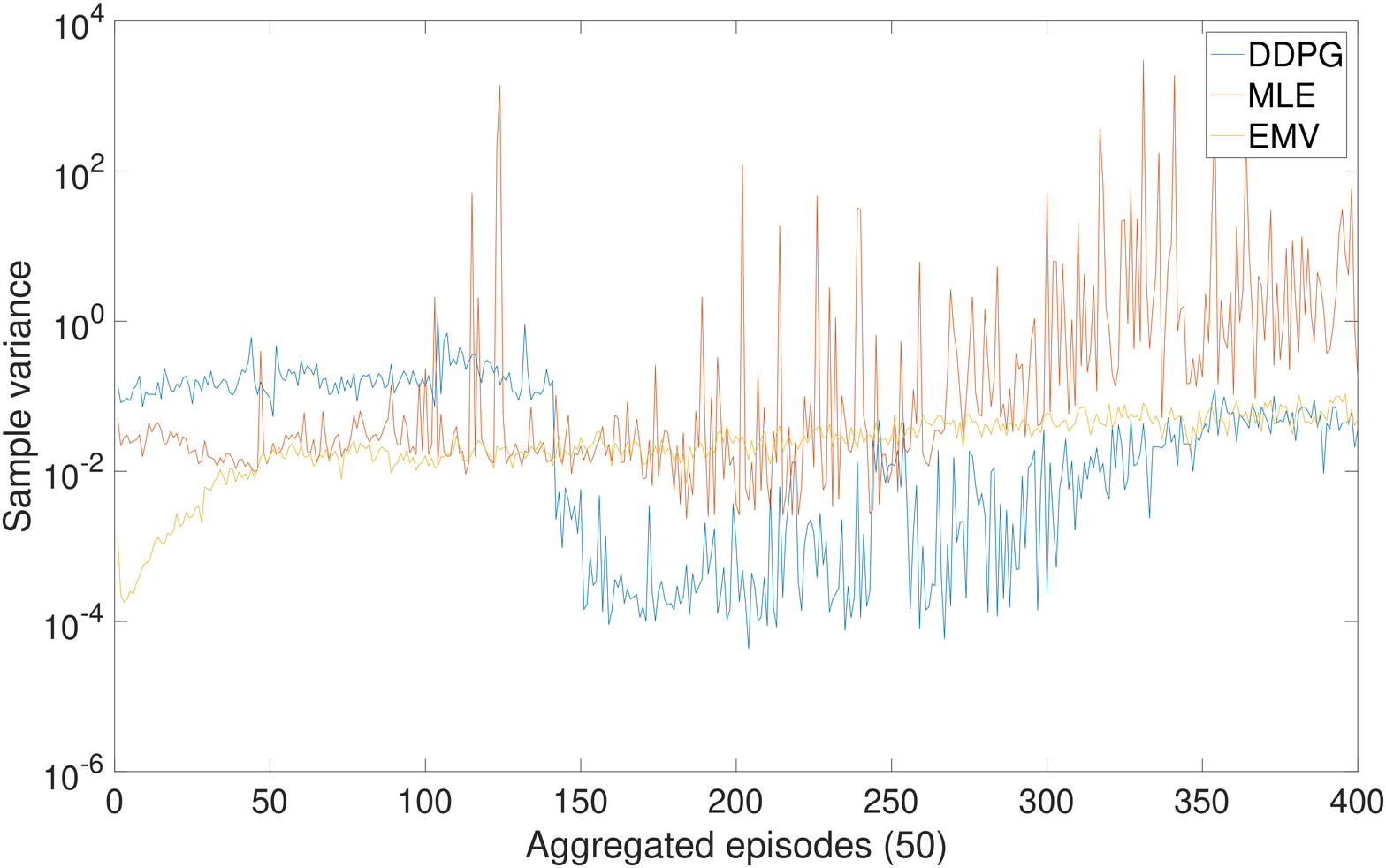}
   \caption{Learning curves of sample variances of terminal wealth (over every $50$ iterations) for EMV, MLE and DDPG for non-stationary market scenario ($\mu_0=-30\%,	\sigma_0=10\%, \delta=0.0001, \gamma =0$).}\label{fig_variance_non}
\end{center}
\end{figure}

\subsection{The decaying exploration case}
A decaying exploration scheme is often desirable for RL since exploitation ought to gradually take more weight over exploration, as more learning iterations have been carried out.\footnote{It is important to distinguish between decaying during  a given learning episode
$[0,T]$ and decaying during the entire learning process (across different episodes). The former has been derived in Theorem 1 as the decay of the Gaussian variance in $t$. The latter refers to the notion that less exploration is needed in later learning episodes.}  Within the current entropy-regularized relaxed stochastic control framework, we have been able to show the convergence from the solution of the exploratory MV to that of the classical MV in Theorem \ref{convergence_to_Dirac}, as the tradeoff parameter $\lambda\rightarrow 0$. Putting together Theorem \ref{convergence_to_Dirac} and Theorem \ref{convergence_learning}, we can reasonably expect a further improvement of the EMV algorithm when it adopts a decaying $\lambda$ scheme, rather than a constant $\lambda$  as in the previous two sections. Our numerical simulation result reported in Figure \ref{fig_hist} demonstrates slightly improved performance of the EMV algorithm with a specifically chosen $\lambda$ process that decreases across all the $M$ episodes, given by
\begin{equation}\label{lambda}
\lambda_k=\lambda_0\left(1-\exp\left(\frac{200(k-M)}{M}\right)\right), \quad \text{for}\ k=0,1,2,\dots, M.
\end{equation}
We plot the histogram of the last $2000$ values of terminal wealth generated by the EMV algorithm with the decaying $\lambda$ scheme starting from $\lambda_0=2$ and the original EMV algorithm with constant $\lambda=2$, respectively. The Sharpe ratio increases from $3.039$ to $3.243$. The result in Figure \ref{fig_hist} corresponds to the stationary market case; the case for non-stationary market demonstrates similar improvement by adopting the decaying exploration scheme.
\begin{figure}[H]
\begin{center}
   \includegraphics[scale=0.5]{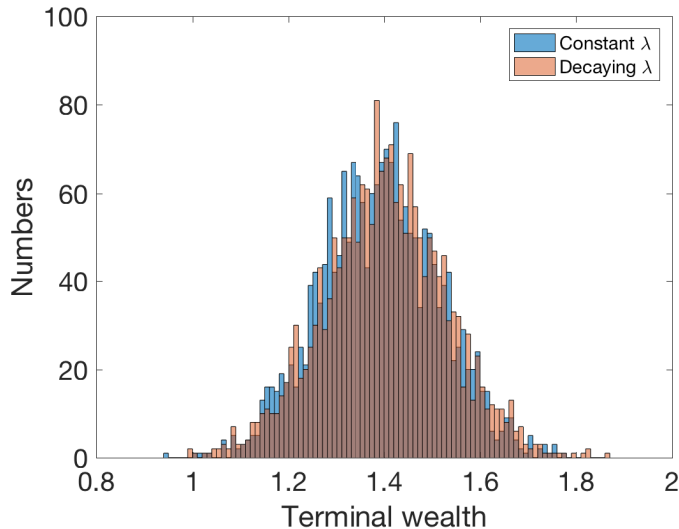}
   \caption{Histogram of the last $2000$ values of terminal wealth for EMV with constant $\lambda$ and decaying $\lambda$ ($\mu=-30\%,	\sigma=10\%$).}\label{fig_hist}
\end{center}
\end{figure}

\section{Conclusions}

In this paper we have developed an RL framework for the continuous-time MV portfolio selection, using the exploratory stochastic control formulation  recently proposed and studied in \cite{Hwang}. By recasting the MV portfolio selection  as an exploration/learning and exploitation/optimization  problem, we are able to derive a data-driven solution, completely skipping any estimation  of the unknown model parameters which is a notoriously difficult, if not insurmountable, task in investment practice. The exploration part is explicitly captured by the relaxed stochastic control formulation and the resulting exploratory state dynamics, as well as the entropy-regularized objective function of the new optimization problem. We prove that the feedback control distribution that optimally balances exploration and exploitation in the MV setting is a Gaussian distribution with a time-decaying variance. Similar to the case of  general LQ problems in \cite{Hwang}, we establish the close connections between the exploratory MV and the classical MV problems, including the solvability equivalence and the convergence as exploration decays to zero.

The RL framework  of the classical MV problem also allows us to design a competitive and interpretable RL algorithm, thanks in addition to a proved  policy improvement theorem and the explicit functional structures for both the value function and the optimal control policy. The policy improvement theorem yields an updating scheme for the control policy that improves the objective  value in each iteration.
The explicit structures of the {\it theoretical} optimal solution to the exploratory MV problem suggest simple yet efficient function approximators without having to resort to  black-box approaches such as neural network approximations. The advantage of our method has been demonstrated  by various numerical simulations with both stationary and non-stationary market environments, where our RL algorithm generally outperforms other two approaches by large margins when solving the MV problem in the continuous control setting.

It should also be noted that  the MV problem considered in this paper is {\it almost} model-free, in the sense that what is essentially needed for the underlying theory
is the LQ structure only, namely, the incremental change in wealth depends linearly  on wealth and
portfolio, and the objective function is quadratic in the terminal wealth.
The former is a reasonable assumption so long as the incremental change in the risky asset price is linear in the price itself (including but not limited to the case when the price is lognormal and the non-stationary case studied in Section 5.2), and the latter is an intrinsic property of
the MV formulation due to the variance term. Therefore, our algorithm is data-driven on one hand yet
entirely interpretable on the other (in contrast to the black-box approach).

Interesting future research directions include empirical study of our RL algorithm using real market data, in the high-dimensional control setting where allocations among multiple risky assets need to be determined at each decision time. For that, our exploratory MV formulation would remain intact  and a multivariate Gaussian distribution would be learned using an algorithm similar to Algorithm 1. It would be interesting to compare the performance of our algorithm with other existing methods, including the Fama-French model (\cite{Fama}) and the distributionally robust method for MV problem (\cite{robust_MV}). Another open question is to design an endogenous, ``optimal'' decaying scheme for the temperature parameter $\lambda$ as learning advances, an essential quantity that dictates the overall level of exploration and bridges the exploratory MV problem with the classical MV problem. These questions are left for further investigations.

\bibliography{bibtex_2}

\end{document}